\documentclass[11pt,a4paper,reqno]{amsart}%
\usepackage[latin1]{inputenc}
\usepackage{mathrsfs}
\usepackage{dsfont}
\usepackage{hyperref}
\usepackage{amsmath}
\usepackage{amssymb}
\usepackage{amsthm}
\usepackage{amsfonts}
\usepackage{amstext}
\usepackage{amsopn}
\usepackage{amsxtra}
\usepackage{mathrsfs}
\usepackage{dsfont}
\usepackage{graphicx}
\usepackage{color}

\newtheorem{theorem}{Theorem}[section]
\newtheorem{proposition}[theorem]{Proposition}
\newtheorem{remark}[theorem]{Remark}
\newtheorem{lemma}[theorem]{Lemma}
\newtheorem{corollary}[theorem]{Corollary}
\newtheorem{definition}[theorem]{Definition}
\newtheorem{example}[theorem]{Example}

\numberwithin{equation}{section}

\newcommand{\ii}{\infty}
\newcommand\R{{\ensuremath {\mathbb R} }}
\newcommand\C{{\ensuremath {\mathbb C} }}
\newcommand\N{{\ensuremath {\mathbb N} }}

\newcommand\1{{\ensuremath {\mathds 1} }}
\newcommand\dGamma{{\rm d}\Gamma}
\newcommand{\<}{\langle}
\renewcommand{\>}{\rangle}
\newcommand\nn{\nonumber}

\renewcommand\phi{\varphi}
\newcommand{\bH}{\mathbb{H}}
\newcommand{\gH}{\mathfrak{H}}
\newcommand{\gS}{\mathfrak{S}}
\newcommand{\gK}{\mathfrak{K}}
\newcommand{\wto}{\rightharpoonup}

\newcommand{\cM}{\mathcal{M}}

\newcommand{\cB}{\mathcal{B}}

\newcommand{\cE}{\mathcal{E}}

\newcommand{\cF}{\mathcal{F}}
\newcommand{\cN}{\mathcal{N}}

\newcommand{\cH}{\mathcal{H}}

\newcommand{\eps}{\epsilon}

\newcommand{\F}{\mathcal{F}}

\renewcommand{\epsilon}{\varepsilon}
\newcommand\pscal[1]{{\ensuremath{\left\langle #1 \right\rangle}}}
\newcommand{\norm}[1]{ \left| \! \left| #1 \right| \! \right| }
\DeclareMathOperator{\tr}{{\rm Tr}}
\DeclareMathOperator{\Tr}{{\rm Tr}}

\renewcommand{\ge}{\geqslant}
\renewcommand{\le}{\leqslant}
\renewcommand{\geq}{\geqslant}
\renewcommand{\leq}{\leqslant}

\renewcommand{\tilde}{\widetilde}
\newcommand{\cHcl}{\cH_{\rm cl}}

\title[Derivation of nonlinear Gibbs measures]{Derivation of nonlinear Gibbs measures from many-body quantum mechanics}

\author[M. Lewin]{Mathieu LEWIN}
\address{CNRS \& Universit\'e Paris-Dauphine, CEREMADE (UMR 7534), Place de Lattre de Tassigny, F-75775 PARIS Cedex 16, France} 
\email{mathieu.lewin@math.cnrs.fr}

\author[P.~T. Nam]{Phan Th\`anh NAM}
\address{IST Austria, Am Campus 1, 3400 Klosterneuburg, Austria} 
\email{pnam@ist.ac.at}

\author[N. Rougerie]{Nicolas ROUGERIE}
\address{Universit\'e Grenoble 1 \& CNRS,  LPMMC (UMR 5493), B.P. 166, F-38042 Grenoble, France}
\email{nicolas.rougerie@grenoble.cnrs.fr}

\date{May 12, 2015}

\begin{document}

\begin{abstract}
We prove that nonlinear Gibbs measures can be obtained from the corresponding many-body, grand-canonical, quantum Gibbs states, in a mean-field limit where the temperature $T$ diverges and the interaction strength behaves as $1/T$. We proceed by characterizing the interacting Gibbs state as minimizing a functional counting the free-energy relatively to the non-interacting case. We then perform an infinite-dimensional analogue of phase-space semiclassical analysis, using fine properties of the quantum relative entropy, the link between quantum de Finetti measures and upper/lower symbols in a coherent state basis, as well as Berezin-Lieb type inequalities. Our results cover the measure built on the defocusing nonlinear Schr\"odinger functional on a finite interval, as well as smoother interactions in dimensions $d\geq2$. 
\end{abstract}

\maketitle

\setcounter{tocdepth}{2}
\tableofcontents

%%%%%%%%%%%%%%%%%%%%%%%%%%%%%%%%%%%%%%%%%%
%%%%%%%%%%%%%%%%%%%%%%%%%%%%%%%%%%%%%%%%%%
\section{Introduction}
%%%%%%%%%%%%%%%%%%%%%%%%%%%%%%%%%%%%%%%%%%
%%%%%%%%%%%%%%%%%%%%%%%%%%%%%%%%%%%%%%%%%%

Nonlinear Gibbs measures have recently become a useful tool to construct solutions to time-dependent nonlinear Schr\"odinger equations with rough initial data, see for instance~\cite{LebRosSpe-88,Bourgain-94,Bourgain-96,Bourgain-97,Bourgain-00,Tzvetkov-08,BurTzv-08,BurThoTzv-10,ThoTzv-10,Suzzoni-11,ColOh-12,OhQua-13}. These are probability measures which are \emph{formally} defined by
 \begin{equation}
d\mu(u)=z^{-1}e^{-\cE(u)}\,du,
\label{eq:mu_intro}
 \end{equation}
where $\cE(u)$ is the Hamiltonian nonlinear energy and $z$ is an infinite normalization factor. The precise definition of $\mu$ will be discussed below. For bosons interacting through a potential $w$, the energy is
\begin{equation}
\cE(u)=\int_\Omega|\nabla u(x)|^2\,dx+\frac{1}{2}\iint_{\Omega\times\Omega} |u(x)|^2\,|u(y)|^2 w(x-y)\,dx\,dy,
\label{eq:def_cE_intro}
\end{equation}
with $\Omega$ a bounded domain in $\R^d$ and with chosen boundary conditions. The cubic nonlinear Schr\"odinger equation corresponds to $w=c\delta_0$, a Dirac delta. For nicer potentials (say, $w\in L^\ii(\R^d)$), the model is often called \emph{Hartree's functional}. The derivation of such functionals from many-body quantum mechanics has a long history, see for example~\cite{BenLie-83,LieYau-87,LieSeiSolYng-05,GreSei-13,LewNamRou-14,LewNamRou-14c} and references therein.

Our purpose here is to prove that the above nonlinear measures $\mu$ in~\eqref{eq:mu_intro} arise naturally from the corresponding \emph{linear} many-particle (bosonic) Gibbs states, in a mean-field limit where the temperature $T\to\ii$ and the interaction intensity is of order $1/T$. We work in the grand-canonical ensemble where the particle number is not fixed, but similar results are expected to hold in the canonical setting, in dimension $d=1$. In this introduction, we will discuss the non-interacting case $w\equiv0$ in any dimension $d\geq1$, but we will consider the interacting case $w\neq0$ (including $w=\delta_0$) mostly in dimension $d=1$, where the measure $\mu$ is better understood. In dimensions $d\geq2$ our result will require strong assumptions on the interaction $w$ which do not include a translation invariant function $w(x-y)$ and we defer the full statement to Section \ref{sec:interacting}. 

In recent papers~\cite{LewNamRou-14,LewNamRou-14c}, we have studied another mean-field regime where the temperature $T$ is fixed\footnote{Equal to $0$ most of the time, but see~\cite[Section 3.2]{LewNamRou-14} for generalization to fixed $T>0$.}. We worked in the canonical ensemble with a given, diverging, number of particles $N\to\ii$ and an interaction of intensity $1/N$. In this case almost all the particles condense on the minimizers of the energy~$\cE$. More precisely, the limit is a measure $\tilde\mu$ which has its support in the set $\cM$ of minimizers for $\cE$ in the unit sphere of $L^2(\Omega)$. The regime we consider here amounts to taking the temperature $T\to\ii$ at the same time as the average number of particles $N\to\ii$, and we end up with the nonlinear Gibbs measure $\mu$. The employed techniques are related to what we have done in~\cite{LewNamRou-14,LewNamRou-14c}, but dealing with the large temperature limit requires several new tools.

In the core of this article, we consider an abstract situation with a nonlinear energy of the form
$$u\mapsto \pscal{u,hu}_\gH+\frac12\pscal{u\otimes u,w\,u\otimes u}_{\gH\otimes_s \gH}$$
on an abstract Hilbert space $\gH$. Here $h>0$ a self-adjoint operator with compact resolvent on $\gH$ and $w\geq 0$ is self-adjoint on the symmetric tensor product $\gH\otimes_s\gH$. In this introduction we describe our results informally, focusing for simplicity on the example of the physical energy $\cE$ in~\eqref{eq:def_cE_intro}, corresponding to $\gH=L^2(\Omega)$, $h=-\Delta$ and $w=w(x-y)$.

\subsubsection*{The nonlinear Gibbs measures} In order to properly define the measure $\mu$, it is customary to start with the non-interacting case $w\equiv0$. Then the formal probability measure
$$d\mu_0(u)=z_0^{-1}e^{-\int_{\Omega}|\nabla u|^2}\,du$$
is an infinite-dimensional gaussian measure. Indeed, since $\Omega$ is a bounded set, let $(\lambda_j)$ and $(u_j)$ be a corresponding set of eigenvalues and normalized eigenfunctions of $-\Delta$ with chosen boundary conditions\footnote{We assume here $\min\sigma(-\Delta)=\min(\lambda_j)>0$ which is for instance the case of the Dirichlet Laplacian. Other boundary conditions (e.g. Neumann or periodic) may be used but then $-\Delta$ must then be replaced everywhere by $-\Delta+C$ with $C>0$.}. Letting $\alpha_j=\pscal{u_j,u}\in\C$ then $\mu_0$ is the infinite tensor product
$$d\mu_0=\bigotimes_{j\geq1}\left( \frac{\lambda_j}{\pi}e^{-\lambda_j|\alpha_j|^2}\,d\alpha_j\right).$$
This measure is well-defined in $L^2(\Omega)$ in dimension $d=1$ only. In higher dimensions, $\mu_0$ lives on negative Sobolev spaces and it is supported outside of $L^2(\Omega)$, which largely complicates the analysis. In this introduction we will most always assume $d=1$ for simplicity (then $\Omega=(a,b)$ is a finite interval), and only discuss the case $d\geq2$ in the end. 

With the non-interacting measure $\mu_0$ at hand, it is possible to define 
\begin{equation}
d\mu(u):=z_r^{-1}e^{-F_{\rm NL}(u)}\,d\mu_0(u) 
 \label{eq:Gibbs_intro}
\end{equation}
where
\begin{equation}
F_{\rm NL}(u)=\frac12\iint_{\Omega\times\Omega} |u(x)|^2\,|u(y)|^2 w(x-y)\,dx\,dy
\label{eq:F_NL_intro}
\end{equation}
is the nonlinear term in the energy, and 
\begin{equation}
z_r=\int_{L^2(\Omega)} e^{-F_{\rm NL}(u)}\,d\mu_0(u)
\label{eq:def_Z_r_intro}
\end{equation}
is the appropriate normalization factor, called the \emph{relative partition function}. In the defocusing case $F_{\rm NL}\geq0$, the number $z_r$ is always well-defined in $[0,1]$. In order to properly define the measure $\mu$, we however need $z_r>0$ and this requires that 
$F_{\rm NL}$ is finite on a set of positive $\mu_0$-measure. This is the case if, for instance, $d=1$ and $0\leq w\in L^\ii(\R^d)$ or $w=\delta_0$.
Note that $z_r$ is formally equal to $z/z_0$, the ratio of the partition functions of the interacting and non-interacting cases. But both $z$ and $z_0$ are infinite, only $z_r$ makes sense.

\subsubsection*{The quantum model and the mean-field limit}
We now quickly describe the quantum mechanical model which is going to converge to the nonlinear Gibbs measure $\mu$. The proper setting is that of Fock spaces and $k$-particle density matrices, but we defer the discussion of these concepts to the next section. We define the $n$-particle Hamiltonian
$$H_{\lambda,n}=\sum_{j=1}^n (-\Delta)_{x_j}+\lambda\sum_{1\leq k<\ell\leq n}w(x_k-x_\ell)$$
which describes a system of $n$ non-relativistic bosons in $\Omega$, interacting via the potential $w$. This operator acts on the bosonic space $\bigotimes_s^nL^2(\Omega)$, that is on $L^2_s (\Omega ^n)$,  the subspace of $L^2(\Omega^n)$ containing the functions which are symmetric with respect to permutations of their variables. The parameter $\lambda$ is used to vary the intensity of the interaction. Here we will take
$$\lambda\sim 1/T$$
with $T$ being the temperature of the sample, which places us in a mean-field regime, as will be explained later.

At temperature $T>0$, in the grand-canonical setting where the number of particles is not fixed but considered as a random variable, the partition function is
$$Z_\lambda(T)=1+\sum_{n\geq1} \tr\left[\exp\left(-\frac{H_{\lambda,n}}{T}\right)\right]$$
where the trace is taken on the space $\bigotimes_s^nL^2(\Omega)$, but we do not emphasize it in our notation. The free-energy of the system is then $ - T \log Z_\lambda(T).$
On the other hand the non-interacting partition function is
$$Z_0(T)=1+\sum_{n\geq1} \tr\left[\exp\left(-\frac{H_{0,n}}{T}\right)\right].$$
When $\Omega$ is bounded, $Z_0(T)$ is finite for every $T>0$. If $0\leq w\in L^\ii(\Omega)$ or if $w=\delta_0$ and $d=1$, then $Z_\lambda(T)$ is also a well-defined number. Both $Z_\lambda(T)$ and $Z_0(T)$ diverge very fast when $T\to\ii$. One of our results (see Theorem~\ref{thm:main} below) says that
\begin{equation}
\boxed{\lim_{\substack{T\to\ii\\ \lambda T\to1}}\frac{Z_\lambda(T)}{Z_0(T)}=z_r}
\label{eq:limit_intro}
\end{equation}
where $z_r$ is the nonlinear relative partition function, defined in~\eqref{eq:def_Z_r_intro} above. The result~\eqref{eq:limit_intro} is in agreement with the intuitive formula $z_r=z/z_0$ explained before. The limit~\eqref{eq:limit_intro} will be valid in a general situation which includes the cases $0\leq w\in L^\ii(\R)$ and $w=\delta_0$ in dimension $d=1$. 

The limit~\eqref{eq:limit_intro} does not characterize the measure $\mu$ uniquely, but in the case $d=1$ we are able to prove convergence of the density matrices of the grand-canonical Gibbs state. Those are obtained by tracing out all variables but a finite number. Namely, for any fixed $k\geq1$, we will show that 
\begin{multline}
\lim_{\substack{T\to\ii\\ \lambda T\to1}}\frac{1}{T^kZ_\lambda(T)}\sum_{n\geq k}\frac{n!}{(n-k)!}\tr_{k+1\to n}\left[\exp\left(-\frac{H_{\lambda,n}}{T}\right)\right]\\
=\int_{L^2(\Omega)} |u^{\otimes k}\rangle\langle u^{\otimes k}|\, d\mu(u),
\label{eq:limit_intro_DM}
\end{multline}
strongly in the trace-class, where $\tr_{k+1\to n}$ is a notation for the partial trace (see~\eqref{eq:def_DM_partial} below).
It will be shown that there can be only one measure $\mu$ for which the limit~\eqref{eq:limit_intro_DM} holds for all $k\geq1$, and our statement is that this measure is the nonlinear Gibbs measure introduced above in~\eqref{eq:Gibbs_intro}.

\subsubsection*{Strategy of proof: the variational formulation}
Our method for proving~\eqref{eq:limit_intro} and~\eqref{eq:limit_intro_DM} is variational, based on Gibbs' principle. The latter states that, for a self-adjoint operator $A$ on a Hilbert space $\gK$,
\begin{equation}\label{eq:intro_Gibbs}
 - \log \Big(\tr_\gK [e^{-A}]\Big) = \inf_{\substack{0\leq M=M^*\\ \tr_\gK M=1}}  \Big\{ \tr [AM] +  \tr \left[ M\log M\right] \Big\} 
\end{equation}
with infimum uniquely achieved by the Gibbs state $M_0=e^{-A}/\tr_\gK e^{-A}$. The reader can simply think of a finite dimensional space $\gK$ in which case $A$ and $M$ are just hermitian matrices.
The two terms in the functional to be minimized are respectively interpreted as the energy and the opposite of the entropy. Using~\eqref{eq:intro_Gibbs} for $A$ and $A+B$ allows us to write
\begin{equation}\label{eq:intro_Gibbs_relative}
 - \log \left(\frac{\tr_\gK [e^{-A-B}]}{\tr_\gK [e^{-A}]}\right) = \inf_{\substack{0\leq M=M^*\\ \tr_\gK M=1}}  \Big\{ \cH(M,M_0) +\tr_{\gK}[BM]\Big\} 
\end{equation}
where
$$\cH(M,M_0)=\tr_\gK\big[M\,(\log M-\log M_0)\big]$$
is the (von Neumann) \emph{quantum relative entropy}. 

In our particular setting,~\eqref{eq:intro_Gibbs} becomes
\begin{equation}\label{eq:intro free ener T}
 -T \log (Z_\lambda(T)) = \inf_{\substack{\Gamma\geq0\\ \tr\Gamma=1}}  \Big\{ \tr [-\Delta \Gamma ^{(1)}] + \lambda \tr [w \Gamma ^{(2)}]+T \tr \left[ \Gamma \log \Gamma \right] \Big\},
\end{equation}
where $\Gamma$ is varied over all trace-class operators on the Fock space and $\Gamma ^{(1)}$, $\Gamma^{(2)}$ are its one-particle and two-particle density matrices respectively (this will be explained in details below). Just as in~\eqref{eq:intro_Gibbs_relative}, this variational principle may be rewritten as
\begin{equation}
-\log\left(\frac{Z_\lambda (T)}{Z_0(T)}\right)=\inf_{\substack{0\leq \Gamma=\Gamma^*\\ \tr\Gamma=1}}\Big\{\cH(\Gamma,\Gamma_0)+\frac{\lambda}{T}\tr \big[\Gamma^{(2)}w\big]\Big\}
\label{eq:variational-Gamma}
\end{equation}
i.e. $-T \log(Z_\lambda(T)/Z_0(T))$ is the free energy of the interacting system, counted relatively to the non-interacting one. The grand-canonical Gibbs state $\Gamma_T$ is the unique state achieving the above infimum .

In a similar manner, we have 
\begin{equation}
-\log(z_r)=\inf_{\substack{\nu\ \text{probability}\\ \text{measure on $L^2(\Omega)$}}}\left\{\cHcl(\nu,\mu_0)+\int_{L^2(\Omega)} F_{\rm NL}(u)\,d\nu(u)\right\}
\label{eq:Gibbs_Hartree}
\end{equation}
where 
$$
\cH_{\rm cl}(\nu,\mu_0):=\int_{L^2(\Omega)} \left(\frac{d\nu}{d\mu_0}\right) \log \left(\frac{d\nu}{d\mu_0}\right) d\mu_0
$$
is the \emph{classical relative entropy} of two measures. Relating the variational principles~\eqref{eq:variational-Gamma} and~\eqref{eq:Gibbs_Hartree} is now a semiclassical problem, reminiscent of situations that are well-studied \emph{in finite dimensional spaces}, see e.g.~\cite{Lieb-73,Simon-80,Gottlieb-05}. The main difficulty here is that we are dealing with the infinite dimensional phase-space $L^2 (\Omega)$. It will be crucial for our method to study the relative free-energy~\eqref{eq:variational-Gamma} and not the original free-energy directly.

Our strategy is based on a measure $\nu$ which can be constructed from the sequence of infinite-dimensional quantum Gibbs states~\cite{Ammari-HDR,AmmNie-08,AmmNie-09,AmmNie-11,Stormer-69,HudMoo-75,LewNamRou-14,Rougerie-cdf}, in the same fashion as semi-classical measures in finite dimension~\cite{ComRob-12}. This so-called \emph{de Finetti measure} has already played a crucial role in our previous works~\cite{LewNamRou-14,LewNamRou-14c} and it will be properly defined in Section~\ref{sec:deFinetti}. The idea is to prove that $\nu$ must solve the variational principle~\eqref{eq:Gibbs_Hartree}. Then, by uniqueness it must be equal to $\mu$ and the result follows.

Roughly speaking, the philosophy is that the grand-canonical quantum Gibbs state behaves as 
$$ \Gamma_T \approx \int \big| \xi (\sqrt{T}u) \big\rangle \big\langle \xi (\sqrt{T} u ) \big|\, d\nu (u)$$
in a suitable weak sense, where $\xi (v)$ is the coherent state in Fock space built on $v$, which has expected particle number $\norm{v} ^2$. In particular, for the reduced density matrices, one should have in mind that
$$ \Gamma_T ^{(k)} \approx  \frac{T ^k}{k!} \int |u  ^{\otimes k} \rangle \langle u ^{\otimes k} |\, d\nu (u).$$

\subsubsection*{Higher dimensions}
As we said the convergence results~\eqref{eq:limit_intro} and~\eqref{eq:limit_intro_DM} will be proved in an abstract situation which includes the defocusing NLS and Hartree cases~\eqref{eq:def_cE_intro} when $d=1$, but not $d\geq2$. When $d\geq2$, several difficulties occur. The first is that the free Gibbs measure $\mu_0$ does not live over $L^2(\Omega)$, but rather over negative Sobolev spaces $H^{-s}(\Omega)$ with $s>d/2-1$. Because the nonlinear interaction term $F_{\rm NL}$ in~\eqref{eq:F_NL_intro} usually does not make any sense in negative Sobolev spaces, the Gibbs measure $\mu$ is also ill-defined and a regularization has to be introduced. 

In this paper we do not consider the problem of renormalizing translation-invariant interactions $w(x-y)$. For simplicity, we assume that $w$ is a smooth enough, say finite rank, operator on $L^2(\Omega)\otimes_s L^2(\Omega)$, which can be thought of as a regularization of more physical potentials. Under this assumption, we are able to prove the same convergence~\eqref{eq:limit_intro} of the relative partition function (see Theorem~\ref{thm:main2} below). However, we are only able to prove the convergence~\eqref{eq:limit_intro_DM} for $k=1$. We in fact prove that the de Finetti measure of the quantum Gibbs state is $\mu$, but the convergence of higher density matrices does not follow easily. The difficulty comes from the divergence of the average particle number which grows much faster than $T$, in contrast to the case $d=1$ where it behaves like $T$. In particular, the density matrices (divided by an appropriate power of $T$) are all unbounded in the trace-class. We believe they are bounded in a higher 
Schatten space, a claim we could prove for any $k$ in the non-interacting case but only for $k=1$ in the interacting case. It is an interesting open problem to extend our result to more physical interactions and to $k\geq2$. 

\subsubsection*{Open problems and the link with QFT}
Nonlinear Gibbs measures have also played an important role in constructive quantum field theory (QFT)~\cite{GliJaf-87,LorHirBet-11,VelWig-73}. By an argument similar to the Feynman-Kac formula, one can write the (formal) grand-canonical partition function of a quantum field in space dimension~$d$, by means of a (classical) nonlinear Gibbs measure in dimension $d+1$, where the additional variable plays the role of time~\cite{Nelson-73,Nelson-73b}. The rigorous construction of quantum fields then sometimes boils down to the proper definition of the corresponding nonlinear measure. This so-called \emph{Euclidean approach to QFT} was very successful for some particular models and the literature on the subject is very vast (see, e.g.,~\cite{Nelson-73,AlbHoe-74,Simon-74,GueRosSim-75} for a few famous examples and~\cite{Summers-12} for a recent review). Like here, the problem becomes more and more difficult when the dimension grows. The main difficulties are to define the measures in the whole space and to 
renormalize the (divergent) physical interactions. We have not yet tried to renormalize physical interactions in our context or to take the thermodynamic limit at the same time as $T\to\ii$. These questions are however important and some tools from constructive QFT could then be useful.

In the present paper, the situation is different from that of QFT, since we derive the $d$-dimensional classical Gibbs measure from the $d$-dimensional quantum problem in a mean-field type limit. Our goal is not to address the delicate construction of the measure for rough interactions, and we always put sufficiently strong assumptions on $w$ in order to avoid any renormalization issue. 

\smallskip

Some specific tools are required in order to put the aforementioned semiclassical intuition on a rigorous basis and we shall need:

 \smallskip

\noindent$\bullet$ To revisit the construction of the de Finetti measures, beyond what has been done in~\cite{AmmNie-08,AmmNie-09,AmmNie-11,LewNamRou-14}. Our approach uses some fine properties of a particular construction of the de Finetti measure, following~\cite{ChrKonMitRen-07,Chiribella-11,Harrow-13,LewNamRou-14b}.

\noindent$\bullet$ To relate de Finetti measures to lower symbols (or Husimi functions) in a coherent state basis, again following~\cite{ChrKonMitRen-07,Chiribella-11,Harrow-13,LewNamRou-14b}. 

\noindent$\bullet$ To prove a version of the semiclassical (first) Berezin-Lieb inequality~\cite{Berezin-72,Lieb-73,Simon-80} adapted to the relative entropy. This uses fundamental properties of the quantum relative entropy. 

 \smallskip

In our paper we always assume that the interaction is repulsive, $w\geq0$, and another open problem is to derive the nonlinear Gibbs measures when $w$ has no particular sign. For instance, in dimension $d=1$ one could obtain the focusing nonlinear Schr\"odinger model of~\cite{CarFroLeb-14,LebRosSpe-88}.

\subsubsection*{Organization of the paper}
In the next section we quickly review the necessary formalism of Fock spaces and density matrices. In Section~\ref{sec:non-interacting}, we consider the non-interacting case $w\equiv0$. After having defined de Finetti measures in Section~\ref{sec:deFinetti}, we turn to the statement of our main result~\eqref{eq:limit_intro_DM} when $w\neq0$ in Section~\ref{sec:interacting}. Sections~\ref{sec:deFinetti_coherent} and~\ref{sec:entropy} are devoted to the practical construction of de Finetti measures using coherent states and a link between the quantum and classical relative entropies. Finally, in Section~\ref{sec:proof}, we provide the proofs of the main theorems.

\subsubsection*{Acknowledgment.} We thank J\"urg Fr\"ohlich, Alessandro Giuliani, Sylvia Serfaty, Anne-Sophie de Suzzoni and Nikolay Tzvetkov for stimulating discussions. Part of this work has been carried out during visits at the \emph{Institut Henri Poincar\'e} (Paris) and the \emph{Institut Mittag-Leffler} (Stockholm). We acknowledge financial support from the \emph{European Research Council} (M.L., ERC Grant Agreement MNIQS 258023) and the \emph{People Programme / Marie Curie Actions} (P.T.N., REA Grant Agreement 291734) under the European Community's Seventh Framework Programme, and from the \emph{French ANR} (Projets NoNAP ANR-10-BLAN-0101 \& Mathosaq ANR-13-JS01-0005-01).

%%%%%%%%%%%%%%%%%%%%%%%%%%%%%%%%%%%%%%%%%%
%%%%%%%%%%%%%%%%%%%%%%%%%%%%%%%%%%%%%%%%%%
\section{Grand-canonical ensemble I}\label{sec:Fock}
%%%%%%%%%%%%%%%%%%%%%%%%%%%%%%%%%%%%%%%%%%
%%%%%%%%%%%%%%%%%%%%%%%%%%%%%%%%%%%%%%%%%%

In this section we quickly describe the grand-canonical theory based on Fock spaces, which is necessary to state our result. We follow here the presentation of~\cite[Sec.~1]{Lewin-11}. Some more involved results on Fock spaces which are used in the core of the paper are described in Section~\ref{sec:Fock2}.

\subsubsection*{Fock space} Let $\gH$ be any fixed separable Hilbert space. To deal with systems with a large number of particles, it is often convenient to work in the (bosonic) Fock space
$$ \F(\gH)= \bigoplus_{n=0}^\infty \bigotimes_s^n \gH= \C \oplus \gH \oplus \big(\gH\otimes_s\gH\big) \oplus ...$$ 
where $\bigotimes_s^n \gH$ is the symmetric tensor product of $n$ copies of $\gH$. In the applications $\gH=L^2(\Omega)$ and $\bigotimes_s^n \gH$ is the subspace of $L^2(\Omega^n)$ containing the functions which are symmetric with respect to exchanges of variables. We always consider $\bigotimes_s^n\gH$ as a subspace of $\bigotimes^n\gH$ and, for every $f_1,...,f_n\in\gH$, we define the symmetric tensor product as
$$f_1\otimes_s\cdots\otimes_s f_n:=\frac{1}{\sqrt{n!}}\sum_{\sigma}f_{\sigma(1)}\otimes\cdots \otimes f_{\sigma(n)}.$$
For functions, we have for instance 
$$(f\otimes g)(x,y)=f(x)g(y)$$
and 
$$(f\otimes_s g)(x,y)=\frac{1}{\sqrt{2}}\left(f(x)g(y)+f(y)g(x)\right).$$

We recall that any orthonormal basis $\{u_i\}$ of $\gH$ furnishes an orthogonal basis $\{u_{i_1}\otimes_s\cdots\otimes_s u_{i_n}\}_{i_1\leq\cdots \leq i_n}$ of $\bigotimes_s^n\gH$. The norms of the basis vectors are 
\begin{equation}
\norm{u_{i_1}\otimes_s \cdots \otimes_s u_{i_k}}^2=m_1! m_2!\cdots.
\label{eq:norm_vectors}
\end{equation}
Here $m_1$ is the number of indices which are equal to $i_1$, that is, $i_1=\cdots =i_{m_1}<i_{m_1+1}$, and $m_2$ is the number of indices that are equal to $i_{m_1+1}$ (if $m_1<n$). The other $m_j$'s are defined similarly.

All the operators in this paper will be implicitly restricted to the symmetric tensor product $\bigotimes^n_s\gH$. For instance, if $A$ is a bounded operator on $\gH$, then $A^{\otimes n}$ denotes the restriction of the corresponding operator, initially defined on $\bigotimes^n\gH$, to the bosonic subspace $\bigotimes^n_s\gH$. It acts as
$$A^{\otimes n} f_1\otimes_s\cdots\otimes_s f_n= (Af_1)\otimes_s\cdots \otimes_s (Af_n).$$
We remark that if $A\geq0$ with eigenvalues $a_i\geq0$, then
\begin{equation}
\tr_{\bigotimes_s^n \gH}\big[A^{\otimes n}\big] \leq \tr_{\bigotimes^n \gH}\big[A^{\otimes n}\big]=\big[\tr(A)\big]^n
\label{eq:trace_power}
\end{equation}
with equality if and only if $A$ has rank 1.\footnote{${\rm Rank}(A)=1$ is also a necessary and sufficient condition for $A^{\otimes n}$ to be completely supported on the bosonic subspace. A reformulation is that the only bosonic factorized states over $\bigotimes^n\gH$ are pure~\cite[Sec.~4]{HudMoo-75}.}
Similarly, for $A_k$ an operator on $\bigotimes^k_s\gH$, we define the operator on $\bigotimes^n_s\gH$
\begin{equation}
A_k\otimes_s\1_{n-k} := {n\choose k}^{-1}\sum_{1\leq  i_1< \cdots< i_k\leq n}(A_k)_{i_1...i_k}
 \label{eq:A_k}
\end{equation}
where $(A_k)_{i_1...i_k}$ acts on the $i_1,..,i_k$th variables. If $A_k\geq0$, then 
\begin{align}
 \tr_{\bigotimes_s^n \gH}\big[A_k\otimes_s\1_{n-k}\big]& \leq \tr_{\bigotimes_s^k \gH}\big[A_k\big]\dim \bigotimes_s^{n-k} \gH \nonumber
 \\&={{n-k+d-1}\choose{d-1}}\tr_{\bigotimes_s^k \gH}\big[A_k\big]
 \label{eq:trace_A_k}
\end{align}
with $d=\dim(\gH)$.

It is useful to introduce the \emph{number operator}
\begin{equation}
\cN:=0 \oplus 1 \oplus 2 \oplus... = \bigoplus_{n=1}^\infty n\; \1_{\bigotimes_s^n\gH}
\end{equation}
which counts how many particles are in the system.

\subsubsection*{States and associated density matrices}
A (mixed) quantum state is a self-adjoint operator $\Gamma \ge 0$ on $\F(\gH)$ with $\tr_{\cF(\gH)} \Gamma=1$.  A state is called \emph{pure} when it is a rank-one orthogonal projection, denoted as $\Gamma=|\Psi\>\<\Psi|$ with $\norm{\Psi}_{\cF(\gH)}^2=1$. In particular, the vacuum state $|0\rangle \langle 0|$ with $|0\rangle= 1\oplus 0 \oplus 0 ...$ is the state with no particle at all.
In principle, states do not necessarily commute with the number operator $\cN$, that is, they are not diagonal with respect to the direct sum decomposition of $\cF(\gH)$. In this paper, we will however only deal with states commuting with $\cN$, which can therefore be written as a (infinite) block-diagonal operator
$$\Gamma=G_0\oplus G_1\oplus\cdots$$
where each $G_n\geq0$ acts on $\bigotimes_s^n\gH$.
For such a quantum state $\Gamma$ on $\F(\gH)$ commuting with $\cN$ and for every $k=0,1,2,...$, we then introduce the $k$-particle density matrix $\Gamma^{(k)}$, which is an operator acting on $\bigotimes_s^k\gH$. A definition of $\Gamma^{(k)}$ in terms of creation and annihilation operators is provided in Section~\ref{sec:Fock2}. Here we give an equivalent definition based on partial traces: 
\begin{equation}
\boxed{\Gamma^{(k)}:=\sum_{n\geq k}{n \choose k} \tr_{k+1\to n}(G_n).}
\label{eq:def_DM_partial}
\end{equation}
The notation $\tr_{k+1\to n}$ stands here for the partial trace associated with the $n-k-1$ variables. A different way to write~\eqref{eq:def_DM_partial} is by duality: 
$$\tr_{\bigotimes_s^k\gH}\big(A_k \Gamma^{(k)}\big)=\sum_{n\geq k} {n \choose k}\tr_{\bigotimes_s^n\gH}(A_k \otimes_s \1_{n-k}G_n)$$
for any bounded operator $A_k$ on $\bigotimes_s^k\gH$, where $A_k \otimes_s \1_{n-k}$ was introduced in~\eqref{eq:A_k}. The density matrices $\Gamma^{(k)}$ are only well-defined under suitable assumptions on the state $\Gamma$, because of the divergent factor $n\choose k$ in the series. For instance, under the assumption that $\Gamma$ has a finite moment of order $k$,
$$\tr_{\cF(\gH)}(\cN^k\Gamma)<\ii,$$
then $\Gamma^{(k)}$ is a well-defined non-negative trace-class operator, with  
\begin{equation}
\tr_{\bigotimes_s^k\gH} \Gamma^{(k)}=\sum_{n\geq k}{n \choose k} \tr(G_n) = \Tr_{\cF(\gH)} \left[ {\cN \choose k} \Gamma \right] \leq\frac{\Tr_{\cF(\gH)} \left[ \cN^k \Gamma \right]}{k!}<\ii.
\label{eq:estim_DM_N_k} 
\end{equation}
We remark that the density matrices are defined in the same manner if $\Gamma$ does not commute with $\cN$, that is, by definition $\Gamma^{(k)}$ does not depend on the off-diagonal blocks of $\Gamma$.

\subsubsection*{Observables} An observable is a self-adjoint operator $\bH$ on $\cF(\gH)$ and the corresponding expectation value in a state $\Gamma$ is $\tr(\bH\Gamma)$. For a pure state $\Gamma=|\Psi\>\<\Psi|$ we find $\pscal{\Psi,\bH\Psi}$.

Several natural observables on the Fock space may be constructed from simpler operators on $\gH$ and $\gH\otimes_s\gH$.
For instance, if we are given a self-adjoint operator $h$ on the one-particle space $\gH$, then we usually denote by
$$\sum_{j=1}^nh_j=n\,h\otimes_s \1_{n-1}$$
the operator which acts on $ \bigotimes_s^n\gH$. On the Fock space $\cF(\gH)$, we can gather all these into one operator
\begin{equation}
\bH_0=0\oplus\bigoplus_{n\geq1}\left(\sum_{j=1}^nh_j\right)
\label{eq:one_particle_Hamiltonians}
\end{equation}
which is often written as $\dGamma(h)$ in the literature (this notation should not be confused with our choice of $\Gamma$ for the quantum state).
Since $\bH_0$ is a sum of operators acting on one particle at a time, it is often called a \emph{one-particle operator}.
A simple calculation shows that the expectation value of any state can be expressed in terms of the one-particle density matrix as
$$\tr_{\cF(\gH)}\left(\bH_0 \Gamma\right)=\tr_\gH \big(h\Gamma^{(1)}\big).$$

Similarly, if we are given a self-adjoint operator $w$ on the two-particle space $\gH\otimes_s\gH$, we may introduce its second-quantization
$$\mathbb{W}:=0\oplus0\oplus\bigoplus_{n\geq2}\left(\sum_{1\leq j<k\leq n}w_{ij}\right)=0\oplus0\oplus\bigoplus_{n\geq2}\left(\frac{n(n-1)}{2}w\otimes_s \1_{n-2}\right)$$
where $w_{ij}$ acts on the $i$th and $j$th particles. The corresponding expectation value can then be expressed in terms of the two-particle density matrix as
$$\tr_{\cF(\gH)}\big(\mathbb{W} \Gamma\big)=\tr_{\gH\otimes_s\gH} \big(w\Gamma^{(2)}\big).$$

\subsubsection*{Gibbs states and relative entropy}
The Gibbs state associated with a Hamiltonian $\bH$ on $\cF(\gH)$ at temperature $T>0$ is defined by
$$\Gamma_T=Z^{-1}\exp(-\bH/T),$$
where
$$Z=\tr_{\cF(\gH)}\big[\exp(-\bH/T)\big]$$
is the corresponding \emph{partition function}.
Of course, this requires to have $\tr_{\cF(\gH)}[\exp(-\bH/T)]<\ii$. Gibbs' variational principle states that $\Gamma_T$ is the unique minimizer of the free energy, that is, solves the minimization problem
$$\inf_{\substack{\Gamma\geq 0\\ \tr_{\cF(\gH)}\Gamma=1}}\Big\{\tr_{\cF(\gH)}\big(\bH\Gamma\big)+T\tr_{\cF(\gH)}\big(\Gamma\log\Gamma\big)\Big\}.$$
Indeed, we have, for any state $\Gamma$, 
\begin{multline*}
\tr_{\cF(\gH)}\big(\bH\Gamma\big)+T\tr_{\cF(\gH)}\big(\Gamma\log\Gamma\big)-\tr_{\cF(\gH)}\big(\bH\Gamma_T\big)-T\tr_{\cF(\gH)}\big(\Gamma_T\log\Gamma_T\big)\\
=T\tr_{\cF(\gH)}\big(\Gamma(\log\Gamma-\log\Gamma_T)\big):=T\cH(\Gamma,\Gamma_T)\geq0
\end{multline*}
which is called the \emph{relative entropy} of $\Gamma$ and $\Gamma_T$. In this paper we will use several important properties of the relative entropy, that will be recalled later in Section~\ref{sec:entropy}. The first is that $\cH(A,B)\geq0$ with equality if and only if $A=B$, which corresponds to the property that $\Gamma_T$ is the unique minimizer of the free energy. A proof of this may be found in~\cite[Thm~2.13]{Carlen-10}, for instance.

\subsubsection*{Free Gibbs states} 
An important role is played by Gibbs states associated with one-particle Hamiltonians $\bH_0$ in~\eqref{eq:one_particle_Hamiltonians}. It is possible to compute the partition function and the density matrices of such states, as summarized in the following well-known

\begin{lemma}[\textbf{Free Gibbs states}]\label{lem:quasi-free}\mbox{}\\
Let $T>0$ and $h>0$ be a self-adjoint operator on $\gH$ such that 
$$\tr_\gH [\exp(-h/T)]<\ii.$$
Then the partition function of the Gibbs state associated with $\bH_0$ in~\eqref{eq:one_particle_Hamiltonians} is 
\begin{equation}
-\log\tr_{\cF(\gH)}\Big[\exp(-\dGamma(h)/T)\Big]= \tr_\gH\Big[\log(1-e^{-h/T})\Big]
\label{eq:partition_quasi_free}
\end{equation}
and the corresponding density matrices are given by 
\begin{equation}\label{eq:DM_quasi_free}
\left[\frac{e^{-\dGamma(h)/T}}{\tr_{\cF(\gH)}\big(e^{-\dGamma(h)/T}\big)}\right]^{(k)}=\left(\frac{1}{e^{h/T}-1}\right)^{\otimes k}.
\end{equation}
\end{lemma}

We recall that $A^{\otimes k}$ is by definition restricted to the symmetric subspace $\otimes_s^k\gH$. The proof uses some algebraic properties of Fock spaces and it is recalled in Appendix~\ref{sec:proof_quasi_free}. Note that the free Gibbs states belong to the class of \emph{quasi-free states} and that a generalization of Lemma~\ref{lem:quasi-free}, Wick's theorem, holds for this larger class of states~\cite{BacLieSol-94}.

%%%%%%%%%%%%%%%%%%%%%%%%%%%%%%%%%%%%%%%%%%
%%%%%%%%%%%%%%%%%%%%%%%%%%%%%%%%%%%%%%%%%%
\section{Derivation of the free Gibbs measures}\label{sec:non-interacting}
%%%%%%%%%%%%%%%%%%%%%%%%%%%%%%%%%%%%%%%%%%
%%%%%%%%%%%%%%%%%%%%%%%%%%%%%%%%%%%%%%%%%%

After these preparations, we start by studying in this section the limit $T\to\ii$ in the non-interacting case $w\equiv0$. The argument is based on a calculation using the properties of quasi-free states in Fock space, which we have recalled in Lemma~\ref{lem:quasi-free}. 

% %%%%%%%%%%%%%%%%%%%%%%%%%%%%%%%
\subsection{Non-interacting Gaussian measure}
Here we review the construction of Gaussian measures. Most of the material from this section is well-known (see for instance~\cite{Skorokhod-74,Bogachev}, \cite{GliJaf-87,VelWig-73,Simon-74} in Constructive Quantum Field Theory, or~\cite{Tzvetkov-08} for the nonlinear Schrödinger equation), but we give some details for the convenience of the reader.

Let $h>0$ be a self-adjoint operator on a separable Hilbert space $\gH$ such that $h$ has compact resolvent. Let $\{\lambda_i\}_i$ be the eigenvalues of $h$ and let $\{u_i\}_i$ be the corresponding eigenvectors. 
We introduce a scale of spaces of Sobolev type defined by 
\begin{equation}
\gH^s:=D(h^{s/2})=\bigg\{u=\sum_{j\geq1} \alpha_j\,u_j\ :\ \|u\|_{\gH^{s}}^2:=\sum_{j\geq1}\lambda_j^s|\alpha_j|^2<\ii\bigg\}\subset \gH,
\end{equation}
for $s\geq0$, and by $\gH^{-s}=(\gH^s)'$ otherwise. The Hilbert space $\gH$($=\gH^0$) can be identified with $\ell^2(\N, \C)$ through the unitary mapping
$u \mapsto \{\alpha_j\}_{j\geq1}$, with $\alpha_j=\langle u_j, u \rangle$. The spaces $\gH^s$ are then isomorphic to weighted $\ell^2$ spaces,
$$\gH^s\simeq \bigg\{\{\alpha_j\}_{j\geq1}\subset\C\ :\ \sum_{j\geq1}\lambda_j^s|\alpha_j|^2<\ii\bigg\},\quad s\in\R,$$
and we will always make this identification for simplicity. Here we implicitly assumed that $\dim\gH=\ii$; in principle $\gH$ could also be finite-dimensional, which is even simpler.

We can now associate with $h$ a Gaussian probability measure $\mu_0$ defined by
\begin{align}\label{eq:def_mu0}
\boxed{d\mu _0(u) := \mathop  \bigotimes \limits_{i = 1}^\infty  \left( \frac{\lambda_i}{\pi}e^{ - {\lambda} _i|\alpha_i|^2}\,d\alpha_i \right).}
\end{align}
Here $d\alpha=d\Re(\alpha)\,d\Im(\alpha)$ is the Lebesgue measure on $\C\simeq\R^2$. 
The formula~\eqref{eq:def_mu0} must be interpreted in the sense that the cylindrical projection $\mu_{0,K}$ of $\mu_0$ on the finite-dimensional space $V_K$ spanned by $u_1,...,u_K$ is the probability measure on $V_K$
$$d\mu_{0,K}(u)=\prod_{i = 1}^K  \left( \frac{\lambda_i}{\pi}e^{ - {\lambda} _i|\alpha_i|^2}\,d\alpha_i \right),$$
for every $K\geq1$. By~\cite[Lemma 1]{Skorokhod-74}, this defines a unique measure $\mu_0$ if and only if the $\mu_{0,K}$'s satisfy the tightness condition
\begin{equation}
\lim_{R\to\ii}\sup_{K}\mu_{0,K}\big(\{u\in V_K\ :\ \|u\|\geq R\}\big)=0.
\label{eq:tightness_mu_0}
\end{equation}
This condition is not always verified in the Hilbert space $\gH$, hence the need to change the norm used in~\eqref{eq:tightness_mu_0}. A simple calculation shows that
$$\mu_{0,K}\big(\{u\in V_K\ :\ \|u\|_{ \gH^{1-p}}\geq R\}\big)\leq R^{-2}\int_{V_K}\norm{u}_{\gH^{1-p}}^2\,d\mu_{0,K}(u)=R^{-2}\sum_{j=1}^K \lambda_j^{-p}.$$
Therefore, if we assume that 
$$\sum_{j\geq1} \frac{1}{\lambda_j^{p}}=\tr_\gH\left(\frac{1}{h^{p}}\right)<\ii,$$
for some $p>0$, the tightness condition~\eqref{eq:tightness_mu_0} is verified in $\gH^{1-p}$ and $\mu_0$ is well defined in this space.

The so-defined measure $\mu_0$ satisfies a zero-one law, in the sense that, for a subspace $\gH^{1-q}\subset\gH^{1-p}$ with $q<p$, we have either $\mu_0(\gH^{1-q})=1$ or $\norm{u}_{\gH^{1-q}}=+\ii$ $\mu_0$-almost surely. Indeed, 
\begin{align}
&\|u\|_{\gH^{1-q}}=+\ii\ \text{$\mu_0$-almost surely}\nn \\
& \qquad\qquad\qquad\Longleftrightarrow \int_{\gH^{1-p}}e^{\epsilon\|u\|^2_{\gH^{1-q}}}d\mu_0(u)=+\ii \text{ for some $\epsilon>0$}\label{eq:Fernique}\\
& \qquad\qquad\qquad\Longleftrightarrow \int_{\gH^{1-p}}\|u\|^2_{\gH^{1-q}}d\mu_0(u)=\tr_\gH (h^{-q})=+\ii\nn
\end{align}
which is called Fernique's theorem. 

In particular, taking $q=0$, we see that the energy is always infinite: $\pscal{u,hu}\equiv+\ii$ $\mu_0$-almost surely. On the other hand, taking $q=1$, we see that the mass $\|u\|_{\gH}=+\ii$, $\mu_0$-almost surely, if and only if $\tr(h^{-1})=+\ii$.

\begin{example}[Laplacian in a bounded domain $\Omega\subset\R^d$~\cite{LebRosSpe-88,Tzvetkov-08}]\label{ex:lap bound dom}\mbox{}\\
Let $\gH=L^2(\Omega)$ with $\Omega$ a bounded open subset of $\R^d$ and $h=-\Delta+C$ with chosen boundary conditions and $C\geq0$ such that $h>0$. Then we have $\tr(h^{-p})<\ii$ for all $p>d/2$, and therefore $\mu_0$ is well-defined on the Sobolev spaces of order $<1-d/2$. The kinetic energy is always infinite: $\int_\Omega|\nabla u|^2=+\ii$ $\mu_0$-almost surely. For $d\geq2$, the mass is also infinite: $\norm{u}_{L^2(\Omega)}=+\ii$, $\mu_0$-almost surely.
\end{example}

\begin{example}[(An)Harmonic oscillator]\label{ex:schro op}\mbox{}\\
For the harmonic-type oscillators, $\gH=L^2(\R^d)$ and $h=-\Delta+|x|^s$, we have $\tr (h^{-p})<\ii$ for all $p>d/2+d/s$. This follows from the Lieb-Thirring inequality of~\cite[Theorem 1]{DolFelLosPat-06} which gives us
$$\tr h^{-p}\leq 2^p\tr(h+\lambda_1)^{-p}\leq \frac{2^p}{(2\pi)^d}\int_{\R^d}\int_{\R^d} \frac{dx\,dk}{\big(|k|^2+|x|^s+\lambda_1\big)^p}<\ii$$
where $\lambda_1>0$ is the first eigenvalue of $h$.
In particular, for the harmonic oscillator $s=2$ we have $\tr (h^{-p})<\ii$ for all $p>d$. When $d=1$, the trace-class case $p=1$, which we will often refer to below, requires $s>2$, which just fails to include the 1D harmonic oscillator.
\end{example}

% %%%%%%%%%%%%%%%%%%%%%%%%%%%%%%%
\subsection{Non-interacting $k$-particle density matrices}
In this section we consider the $k$-particle density matrices of the Gibbs state described by the measure $\mu_0$. These are formally defined by
\begin{equation}
\boxed{\gamma_0^{(k)}:=\int_{\gH^{1-p}} |u^{\otimes k}\rangle\langle u^{\otimes k}|\;d\mu_0(u).}
\label{eq:def_DM_free}
\end{equation}
Note that each $|u^{\otimes k}\rangle\langle u^{\otimes k}|$ is bounded from $\bigotimes_s^k\gH^{p-1}$ to $\bigotimes_s^k\gH^{1-p}$, with corresponding norm $\|u\|_{\gH^{1-p}}^{2k}$. For a $p\geq1$ such that $\tr(h^{-p})<\ii$, by Fernique's theorem~\eqref{eq:Fernique}, the measure $\mu_0$ has an exponential decay in the Hilbert space $\gH^{1-p}$ and therefore the integral~\eqref{eq:def_DM_free} is convergent in those spaces. Note that each $|u^{\otimes k}\rangle\langle u^{\otimes k}|$ can be unbounded $\mu_0$-almost surely on $\bigotimes_s^k\gH$ (if $p>1$ and $\tr(h^{-1})=+\ii$). The following says that, after averaging with the measure $\mu_0$, the resulting operator $\gamma^{(k)}_0$ is actually always compact (hence bounded) on the original Hilbert space $\bigotimes_s^k\gH$. 

\begin{lemma}[\textbf{Density matrices of the non-interacting Gibbs state}]\label{lem:DM_free}
Let $h\geq0$ be a self-adjoint operator on $\gH$ such that $\tr_{\gH}\big(h^{-p}\big)<\ii$ for some $1\leq p<\ii$. 
For every $k\geq1$, we have
\begin{equation}
\boxed{\gamma_0^{(k)}=k!\,(h^{-1})^{\otimes k}.}
\label{eq:DM_free_lemma}
\end{equation}
In particular, $\gamma^{(k)}_0$ extends to a unique compact operator on $\bigotimes_s^k\gH$, with
$$\tr_{\bigotimes_s^k\gH}\big(\gamma_0^{(k)}\big)^p\leq (k!)^{p}\big[\tr_\gH(h^{-p})\big]^k<\ii.$$
\end{lemma}

This lemma tells us that the $k$-particle density matrices of the Gibbs state $\mu_0$ are all compact operators on (tensor products of) the original Hilbert space $\gH$, even if the measure $\mu_0$ itself lives over the bigger space $\gH^{1-p}$. 

\begin{proof} Assume that $h$ has eigenvalues $\lambda_1,\lambda_2,...$ with the corresponding eigenfunctions $u_1,u_2,...$ We claim that 
\begin{align}
&\int |u^{\otimes k}\rangle\langle u^{\otimes k}|\;d\mu_0(u)\nn\\
&\qquad =k!\sum_{i_1\leq i_2\leq\cdots\leq i_k}\left(\prod_{\ell=1}^k\frac{1}{\lambda_{i_\ell}}\right)\frac{|u_{i_1}\otimes_s \cdots \otimes_s u_{i_k}\rangle\langle u_{i_1}\otimes_s \cdots \otimes_s u_{i_k}|}{\norm{u_{i_1}\otimes_s \cdots \otimes_s u_{i_k}}^2}\nn\\
&\qquad =k!\,(h^{-1})^{\otimes k}.
\label{eq:formula_DM_free}
\end{align}
We recall that $(u_{i_1}\otimes_s \cdots \otimes_s u_{i_k})_{i_1\leq i_2\leq\cdots \leq i_k}$ forms an orthogonal basis of $\bigotimes_s^k\gH$. The proof that the second line in~\eqref{eq:formula_DM_free} equals the third can be done by simply applying $h^{\otimes k}$ on the right, which gives $k!$ times the identity. The bound on the trace then follows from~\eqref{eq:trace_power}.

Let us now derive the first line in~\eqref{eq:formula_DM_free}. Each $u_{i_1}\otimes_s \cdots \otimes_s u_{i_k}$ belongs to $\bigotimes_s^k\gH^p$ for every $p$, and therefore we can compute its expectation with the $k$-particle density matrix $\gamma_0^{(k)}$. Using then that 
$$\frac{\int_{\C}|\alpha|^{2m}e^{-\lambda|\alpha|^2}d\alpha}{\int_{\C}e^{-\lambda|\alpha|^2}d\alpha}=\frac{\int_0^\ii r^me^{-\lambda r}dr}{\int_0^\ii e^{-\lambda r}dr}=\frac{m!}{\lambda^m}$$
and introducing the same integers $m_1,m_2,...$ as in~\eqref{eq:norm_vectors}, we find
\begin{align*}
&\pscal{u_{i_1}\otimes_s\cdots \otimes_s u_{i_k},\left(\int_{\gH}|u^{\otimes k}\rangle\langle u^{\otimes k}|\, d\mu_0(u)\right)u_{i_1}\otimes_s\cdots \otimes_s u_{i_k}}\\
&\qquad =k!\int\prod_{\ell=1}^k|\pscal{u,u_{i_\ell}}|^2 d\mu_0(u)=\frac{k!\,m_1!\cdots}{\prod_{\ell=1}^k\lambda_{i_\ell}}=\frac{k!\,\norm{u_{i_1}\otimes_s\cdots \otimes_s u_{i_k}}^2}{\prod_{\ell=1}^k\lambda_{i_\ell}}.
\end{align*}
If we put different indices on the right and on the left, then we get 0. Using the orthogonality of the $u_{i_1}\otimes_s\cdots \otimes_s u_{i_k}$, this ends the proof of~\eqref{eq:formula_DM_free}. 
\end{proof}

%%%%%%%%%%%%%%%%%%%%%%%%%%%%%%%
\subsection{High-temperature limit}

In this section we explain how the Gibbs measure $\mu_0$ and the density matrices $\gamma_0^{(k)}$ defined above in~\eqref{eq:def_DM_free} arise in the high-temperature limit of the grand-canonical many-body quantum system. We define the non-interacting Hamiltonian in Fock space $\cF(\gH)$ by 
$$ \bH_0=0 \oplus \bigoplus_{n=1}^\infty \left(\sum_{i=1}^n h_i\right).$$
The corresponding Gibbs state is
\begin{equation}
\label{eq:Gibb-non-int}
\Gamma_{0,T}=Z_{0}(T)^{-1}\,\exp\left(-\frac{\bH_0}{T}\right)
\end{equation}
where
\begin{equation}\label{eq:Z-non-int}
Z_{0}(T)=\tr_{\cF(\gH)} \left[\exp\left(-\frac{\bH_0}{T}\right)\right]=1+\sum_{n\geq1}\tr_{\bigotimes_s^n\gH} \;\exp\left(-\frac{\sum_{j=1}^nh_j}{T}\right)
\end{equation}
is the associated partition function. The state $\Gamma_{0,T}$ is quasi-free and Lemma~\ref{lem:quasi-free} tells us that
\begin{equation}
 -\log Z_{0}(T)= \tr_\gH \log(1-e^{-h/T}),
 \label{eq:formula_Z_0}
\end{equation}
which is finite if and only if $\tr_\gH e^{-h/T} <\infty$. We always assume that 
\begin{equation}
\tr_{\gH}\big(h^{-p}\big)<\ii
\label{eq:assumption_h}
\end{equation}
for some $1\leq p<\ii$ which of course implies $\tr_\gH e^{-h/T} <\infty$ for all $T>0$. Lemma~\ref{lem:quasi-free} also gives us the formula
\begin{equation} \label{eq:formula_DM_free_T}
\Gamma_{0,T}^{(k)}=\left(\frac{1}{e^{h/T}-1}\right)^{\otimes k}
\end{equation}
for the $k$-particle density matrix of $\Gamma_{0,T}$. 

We recall that the Schatten space $\gS^p(\gK)$ of a Hilbert space $\gK$ is
$$\gS^p(\gK):=\big\{A:\gK\to\gK\ : \ \norm{A}_{\gS^{p}(\gK)}^p:=\tr |A|^{p}=\tr (A^*A)^{p/2}<\ii\big\}.$$

\begin{lemma}[\textbf{Convergence to Gaussian measures}]\label{lem:CV_free}\mbox{}\\
Let $h\geq0$ be a self-adjoint operator on $\gH$ such that $\tr_{\gH}\big(h^{-p}\big)<\ii$ for some $p\geq1$.
Let $\mu_0$ be the Gaussian measure defined in~\eqref{eq:def_mu0} with $k$-particle density matrix $\gamma^{(k)}_0$ defined in~\eqref{eq:def_DM_free}. Then, we have
$$\boxed{\frac{k!}{ T^k }\Gamma_{0,T}^{(k)}\underset{T\to\ii}\longrightarrow \gamma_0^{(k)}=\int_{\gH^{1-p}} |u^{\otimes k}\rangle\langle u^{\otimes k}|\,d\mu_0(u)=k!\,(h^{-1})^{\otimes k}}$$
strongly in the Schatten space $\gS^{p}(\bigotimes_s^k\gH)$, for every fixed $k\geq1$. Moreover, the number of particles in the system behaves as
\begin{equation} \label{eq:cv-number-particle}
\lim_{T\to\ii}\frac{\tr_{\cF(\gH)}\left(\cN\Gamma_{0,T}\right)}{T}=\tr_\gH\, h^{-1}\leq \ii.
\end{equation}
\end{lemma}

\begin{proof} 
We have
$$\frac{1}{T(e^{h/T}-1)}\leq h^{-1}$$
and the proof follows immediately from the dominated convergence theorem in Schatten spaces~\cite[Thm 2.16]{Simon-79}.
\end{proof}

Due to the definition~\eqref{eq:def_DM_partial} of the density matrix $\Gamma^{(k)}_{0,T}$, the convergence is the result~\eqref{eq:limit_intro_DM} that was claimed in the introduction, in the non-interacting case. In concrete models, all physical quantities may be expressed in terms of the $\Gamma^{(k)}$s. For instance, when $\gH=L^2(I)$ with $I$ a bounded interval and $h=-d^2/dx^2$ with Dirichlet boundary conditions, then we obtain for the free Bose gas in $\Omega$ convergence of the  density  
$$\frac{\Gamma^{(1)}_{0,T}(x,x)}{T}\underset{T\to\ii}\longrightarrow \int_{L^2(I)} |u(x)|^2\,d\mu_0(u),$$
and of the density of kinetic energy
$$\frac{\widehat{\Gamma^{(1)}_{0,T}}(p,p)}{T}\underset{T\to\ii}\longrightarrow \int_{L^2(I)} |\widehat{u}(p)|^2\,d\mu_0(u),$$
strongly in $L^1(I)$. These quantities can be measured in experiments. The claimed convergence follows from the continuity of the linear map $\gamma\in\gS^1(L^2(I))\mapsto \gamma(x,x)\in L^1(I)$.

%%%%%%%%%%%%%%%%%%%%%%%%%%%%%%%%%%%
%%%%%%%%%%%%%%%%%%%%%%%%%%%%%%%%%%%
\section{De Finetti measures}\label{sec:deFinetti}
%%%%%%%%%%%%%%%%%%%%%%%%%%%%%%%%%%%
%%%%%%%%%%%%%%%%%%%%%%%%%%%%%%%%%%%

It is possible to reformulate the previous result in terms of \emph{de Finetti measures}. We give a definition here, as we will need it to state our main results for the interacting model in the next section. All the technical details will be provided in Section~\ref{sec:deFinetti_coherent} below.

As usual for problems settled in Fock space, the argument is based on \emph{coherent states} which are defined for every $u\in \gH$ by
\begin{equation}
\xi(u):=e^{-\frac{|u|^2}{2}} \bigoplus_{j \ge 0} \frac{1}{\sqrt{j!}} u^{\otimes j}\in\cF(\gH).
\label{eq:def_coherent}
\end{equation}
Here $u\in\gH$ is not necessarily normalized in $\gH$, but the exponential factor makes $\xi(u)$ a state on $\cF(\gH)$. An important tool is the resolution of the identity on any finite-dimensional subspace $V\subset\gH$:
\begin{equation}
\int_{V} |\xi(u)\rangle  \langle \xi(u)| du = \left( \int_{V} e^{-|u|^2} du \right) \1_{\F(V)} = \pi^{\dim(V)}  \1_{\F(V)}.
\label{eq:resolution_coherent}
\end{equation}
which follows from rotational invariance of the normalized uniform measure $du$ on the sphere of $V$ (Schur's lemma). Here $\cF(V)$ is identified to a subspace of $\cF(\gH)$ and $\1_{\cF(V)}$ is the associated orthogonal projection. This formula is the starting point of all the following arguments. For any given sequence $0<\eps_n\to0$ (playing the role of a semi-classical parameter), we define the \emph{anti-Wick quantization} of a function $b\in C^0_b(V)$ at scale $\eps_n$, with $V$ an arbitrary finite-dimensional subspace of $\gH$, by
\begin{equation}
\mathbb{B}_{\eps_n}:=(\eps_n\pi)^{-\dim(V)}\int_{V} b(u)\;|\xi(u/\sqrt{\eps_n})\>\<\xi(u/\sqrt{\eps_n})|\, du.
\label{eq:def_quantization_b}
\end{equation}

The de Finetti measure of a sequence of states is then obtained by looking at the weak limits against the anti-Wick quantization of any function $b$, similarly to the semi-classical measures in finite-dimensional semi-classical analysis (see, e.g.,~\cite[Sec.~2.6.2]{ComRob-12}).

\begin{definition}[de Finetti measures~\cite{AmmNie-08}]\mbox{}\\
Let $h>0$ be a self-adjoint operator with compact resolvent and let $\gH^s$ be the scale of Sobolev spaces defined in Section~\ref{sec:non-interacting}.
Let $\{\Gamma_n\}$ be a sequence of states on the Fock space $\cF(\gH)$ (that is, $\Gamma_n\geq0$ and $\tr[\Gamma_n]=1$) and $0<\epsilon_n\to0$. We say that a measure $\nu$ on $\gH^{1-p}$ with $p\geq1$ is the de Finetti measure of the sequence $\{\Gamma_n\}$ at scale $\eps_n$, if we have
\begin{equation}
\lim_{n\to\ii}\tr\big[\mathbb{B}_{\eps_n}\Gamma_n\big]=\int_{\gH^{1-p}}b(u)\,d\nu(u)
\label{eq:limit_Anti-Wick}
\end{equation}
for every finite-dimensional  subspace $V\subset\gH$ and every $b\in C^0_b(V)$.
\end{definition}

When $\nu$ exists, then it is \emph{unique}, since a measure in a Hilbert space is characterized by its expectation against bounded continuous functions living over an arbitrary finite-dimensional subspace~\cite{Skorokhod-74}. Also, due to the uniform bound
\begin{equation}
\norm{\mathbb{B}_{\eps_n}}\leq \norm{b}_{L^\ii(V)},
\label{eq:unif_bound_Anti_Wick}
\end{equation}
a density argument shows that it suffices to require~\eqref{eq:limit_Anti-Wick} for all $b\in C^0_b(V_J)$ with $V_J={\rm span}(u_1,...,u_J)$ and $\{u_j\}$ a chosen orthonormal basis of $\gH$ (which, in our case, will always be a basis of eigenvectors of the one-particle operator~$h$).

The above definition of $\nu$ based on coherent states is taken from~\cite{AmmNie-08}. The anti-Wick quantization is not the only possible choice and we refer to~\cite{AmmNie-08} for a discussion of the \emph{Weyl quantization}. Another definition related to the \emph{Wick quantization} relies on density matrices and requires that 
\begin{equation}\label{eq:deF Wick}
k!\,(\eps_n)^k\,\Gamma_n^{(k)}\wto \int_{\gH^{1-p}}|u^{\otimes k}\>\<u^{\otimes k}|\,d\nu(u) 
\end{equation}
weakly in $\otimes_s^k\gH^{p-1}$ for all $k\geq1$, where the weak convergence means that 
$$k!(\epsilon_{n_j})^k\<\Psi,\Gamma_{n_j}^{(k)}\Psi\>\longrightarrow \int_{\gH^{1-p}} |\<u^{\otimes k},\Psi\>|^2 \,d\nu(u),\qquad\forall \Psi\in \otimes_s^k\gH^{p-1}.$$ 
The latter approach was used in~\cite{LewNamRou-14} and in the previous section. We remark that the convergence \eqref{eq:deF Wick} of {\em all} density matrices is sufficient to characterize the limiting measure $\nu$ uniquely since, by an argument of~\cite[Section 2]{LewNamRou-14}, a measure $\nu$ is determined completely by {\em all} of its moments $\int_{\gH^{1-p}} |u^{\otimes k}\rangle \langle u^{\otimes k}| \,d\nu(u)$.

When it applies, the Wick approach~\eqref{eq:deF Wick} is useful to understand the basic principle at work. In particular, in the trace class case $p=1$, it suggests that the expectation value of $\cN ^k$ in $\Gamma_n$ behaves as $(\eps_n)^{-k} \to \infty$ when $n\to \infty$, so that we are dealing with states essentially living on sectors containing $O(\eps_n ^{-1})\to \infty$ bosonic particles. In the general case $p\ge 1$, the expectation value of $\cN ^k$ in $\Gamma_n$ may grow much faster than $(\eps_n)^{-k}$, but the existence of the de Finetti measure is nevertheless a fundamental consequence of the large size of the system under consideration~\cite{ChrKonMitRen-07,HudMoo-75,Rougerie-cdf,Stormer-69} and it is therefore reasonable to expect it to hold for our states in the limit $n\to \infty$.

\medskip

Because the density matrices are not always bounded for an arbitrary sequence of states, an appropriate control is then needed. The Wick quantization requires some bounds on {\em all} density matrices. When $p>1$ these bounds are difficult to prove and it will therefore be more convenient to base our arguments on the anti-Wick quantization. In the same spirit as in~\cite[Thm~6.2]{AmmNie-08} and~\cite[Thm~2.2]{LewNamRou-14}, we are able to prove that any sequence $\{\Gamma_n\}$ satisfying suitable estimates has a de Finetti measure $\nu$ after extraction of a subsequence.

\begin{theorem}[\textbf{Existence of de Finetti measures}]\label{thm:deFinetti}\mbox{}\\
Let $h>0$ be a self-adjoint operator with compact resolvent and let $\gH^s$ be the scale of Sobolev spaces defined in Section~\ref{sec:non-interacting}.
Let $0\leq \Gamma_n$ be a sequence of states on the Fock space $\cF(\gH)$, with $\tr_{\cF(\gH)}\Gamma_n=1$. Assume that there exists a real number $1\leq \kappa\leq\ii$ and a sequence $0<\epsilon_n\to0$ such that
\begin{equation}
\tr\Big[\big(\epsilon_n \dGamma(h^{1-p})\big)^{s}\,\Gamma_n\Big]\leq C_s\ <\infty
\label{eq:bound_DM}
\end{equation}
for some $1\leq p<\ii$, and for $s=\kappa$ if $\kappa<\ii$ (for all $1\leq s<\kappa$ if $\kappa=+\ii$). 

Then there exists a Borel probability measure $\nu$ on $\gH^{1-p}$ (invariant under multiplication by a phase factor) which is the de Finetti measure of a subsequence $\Gamma_{n_j}$ at scale $\eps_{n_j}$. More precisely, 
\begin{equation}
\lim_{n_j\to\ii}\tr_{\cF(\gH)}\big(\Gamma_{n_j}\mathbb{B}_{\epsilon_{n_j}}\big)=\int_{\gH^{1-p}} b(u)\,d\nu(u)
\label{eq:weak_limit_Hartree_quantization}
\end{equation}
for all $b\in C^0_b(V)$ with $V$ an arbitrary finite-dimensional subspace of $\gH$ (here $\mathbb{B}_{\epsilon_n}$ is the anti-Wick quantization of $b$ that was introduced above in~\eqref{eq:def_quantization_b}). Moreover,  
\begin{equation}
k!(\epsilon_{n_j})^k\,\Gamma_{n_j}^{(k)}\wto \int_{\gH^{1-p}} |u^{\otimes k}\rangle \langle u^{\otimes k}| \,d\nu(u)
\label{eq:weak_limit_DM}
\end{equation}
weakly in $\bigotimes_s^k\gH^{p-1}$ for every integer $1\leq k< \kappa$.
\end{theorem}

The weak convergence~\eqref{eq:weak_limit_DM} in the statement means that 
\begin{equation}\label{eq:weak CV deF}
k!(\epsilon_{n_j})^k\<\Psi,\Gamma_{n_j}^{(k)}\Psi\>\longrightarrow \int_{\gH^{1-p}} |\<u^{\otimes k},\Psi\>|^2 \,d\nu(u),\qquad\forall \Psi\in \otimes_s^k\gH^{p-1}. 
\end{equation}

This theorem generalizes existing results in several directions. First, we emphasize that the moment bound~\eqref{eq:bound_DM} is only assumed to hold for $1\leq s\leq \kappa$ (here $\kappa$ does not have to be an integer) and that the limit for the density matrices is then only valid \emph{a priori} for $1\leq k<\kappa$. Note that if $\kappa<\infty$, then the convergence \eqref{eq:weak_limit_DM} does not characterize the measure $\nu$ uniquely, but the anti-Wick quantization \eqref{eq:weak_limit_Hartree_quantization} always does. When $p=1$ the operator $h$ plays no role in the statement and the result is~\cite[Thm. 6.2]{AmmNie-08} when $\kappa=+\ii$. When $p>1$, the bound involves $\dGamma(h^{1-p})$ instead of the number operator $\cN$ and the final measure lives over $\gH^{1-p}$ instead of $\gH$. The proof nevertheless goes along the lines of~\cite{AmmNie-08,LewNamRou-14} and it will be provided in Section~\ref{sec:proof_de_Finetti} below.

In particular, in the non-interacting case we are able to reformulate the results of the previous section, using Lemma~\ref{lem:CV_free} and Theorem~\ref{thm:deFinetti}:

\begin{corollary}[\textbf{de Finetti measure in the non-interacting case}]\label{cor:deF non int}\mbox{}\\
Let $h\geq0$ be a self-adjoint operator on $\gH$ such that $\tr_{\gH}\big(h^{-p}\big)<\ii$ for some $p\geq1$.
Then $\mu_0$ (the Gaussian measure defined in~\eqref{eq:def_mu0}) is the (unique) de Finetti measure of the sequence of quantum Gibbs states $(\Gamma_{0,T})$ at scale $1/T$.
\end{corollary}

\begin{proof}
We have proved the convergence of {\em all} density matrices in Lemma~\ref{lem:CV_free}. As we remarked before, this determines the de Finetti measure uniquely.
\end{proof}

%%%%%%%%%%%%%%%%%%%%%%%%%%%%%%%%%%%%%%%%%%
%%%%%%%%%%%%%%%%%%%%%%%%%%%%%%%%%%%%%%%%%%
\section{Derivation of the nonlinear Gibbs measures: statements}\label{sec:interacting}
%%%%%%%%%%%%%%%%%%%%%%%%%%%%%%%%%%%%%%%%%%
%%%%%%%%%%%%%%%%%%%%%%%%%%%%%%%%%%%%%%%%%%

In this section we state our main result concerning the high temperature limit of \emph{interacting} quantum particles and the occurrence of the nonlinear Gibbs measure. 

As before we fix a self-adjoint operator $h>0$ on a separable Hilbert space $\gH$, such that $\tr(h^{-p})<\ii$ for some $1\leq p<\ii$. We use the same notation $\gH^s$ as in the previous section for the Sobolev-like spaces based on $h$. In particular, we consider the non-interacting Gaussian measure $\mu_0$ on $\gH^{1-p}$ and the corresponding $k$-particle density matrices $\gamma^{(k)}_0$ which are given by~\eqref{eq:DM_free_lemma} and belong to the Schatten space $\gS^p(\bigotimes_s^k\gH)$, by Lemma~\ref{lem:DM_free}. 

Now we state our main theorems and, for clarity, we first discuss the simpler case $p=1$ which, for concrete models involving the Laplacian, corresponds to space dimension $d=1$.

\subsection{The trace-class case}

In this section we assume that 
$$\tr(h^{-1})<\ii,$$
which implies in particular that the measure $\mu_0$ lives over the original Hilbert space $\gH$. We then consider a non-negative self-adjoint operator $w$ on $\gH\otimes_s\gH$ such that 
\begin{equation}
\int_{\gH} \pscal{u\otimes u,w \, u\otimes u}\,d\mu_0(u)=\tr_{\bigotimes_s^2\gH}\left[w\,h^{-1}\otimes h^{-1}\right]<\ii.
\label{eq:assumption_w}
\end{equation}
It is for example sufficient to assume that $w$ is bounded.

\begin{example}[Laplacian in 1D]\label{ex:1D Lap}\mbox{}\\
When $\Omega=(0,\pi)$ is a bounded interval in $\R$ and $h=-d^2/dx^2$ with Dirichlet boundary conditions (recall Example~\ref{ex:lap bound dom}), the assumption~\eqref{eq:assumption_w} is satisfied for 
$$[w\psi](x,y)=W(x-y)\psi(x,y)$$ 
with $W\in L^p((-\pi,\pi),\R_+)$ and $1\leq p\leq\ii$ or, more generally, $W=W_1+W_2$ with $W_1$ a positive measure with finite mass on $(-\pi,\pi)$ and $W_2\in L^\infty((-\pi,\pi),\R_+)$. Indeed,
\begin{multline}
\tr_{L^2(\Omega\times\Omega)}\left[W(x-y)(-d^2/dx^2)^{-1}(-d^2/dy^2)^{-1}\right]\\=\int_0^\pi\int_0^\pi W(x-y) G(x,x)G(y,y)\,dx\,dy
\label{eq:test_assumption_w}
\end{multline}
where 
$$G(x,y)=\frac{2}{\pi}\sum_{n\geq1}\frac{1}{n^2}\sin(nx)\sin(ny)$$
is the integral kernel of $(-d^2/dx^2)^{-1}$. From this formula it is clear that $x\mapsto G(x,x)$ is in $L^1(0,\pi)\cap L^\ii(0,\pi)$, and~\eqref{eq:test_assumption_w} is finite under the previous assumptions on $W$. The delta function $W=\delta_0$ is allowed, which leads to the Gibbs measure built on the defocusing non-linear Schr\"odinger functional. The situation is of course exactly the same in an arbitray bounded interval $\Omega=(a,b)$ in $\R$, with any other boundary condition, but $-d^2/dx^2$ then has to be replaced by $-d^2/dx^2+C$ in order to ensure $h>0$. 
\end{example}

\begin{example}[(An)Harmonic oscillator in 1D]\label{ex:1D Harm}\mbox{}\\
We have mentioned in Example~\ref{ex:schro op} that $h = -d^2/dx^2+ |x| ^s$  on $L^2(\R)$ satisfies 
$$\tr(h^{-1})<\ii \mbox{ for } s>2.$$
As in the previous example, the assumption~\eqref{eq:assumption_w} is then satisfied for $W\in L^p(\R,\R_+)$ with $1\leq p\leq\ii$ or, more generally, $W=W_1+W_2$ with $W_1$ a positive measure with finite mass on $\R$ and $W_2\in L^\infty(\R,\R_+)$. The argument is the same as in Example~\eqref{ex:1D Lap} with, this time,
$$G(x,y):=\big(-d^2/dx^2+ |\cdot|^s\big)^{-1}(x,y).$$
We have 
$$\int_\R G(x,x)\,dx=\tr h^{-1}<\ii$$
and, using that 
$$-\frac{d^2}{dx^2}+ |x| ^s\geq \epsilon\left(-\frac{d^2}{dx^2}+1\right)$$
for all $\epsilon\leq\lambda_1/(1+\lambda_1)$ where $\lambda_1$ is the first eigenvalue of $h$, we deduce 
$$G(x,x)\leq \frac{1}{2\pi\epsilon}\int_\R \frac{dk}{k^2+1}<\ii.$$
Therefore $x\mapsto G(x,x)$ is in $L^1(\R)\cap L^\ii(\R)$ and the conclusion follows.
\end{example}

We remark that~\eqref{eq:assumption_w} implies that $\pscal{u\otimes u,w u\otimes u}\geq0$ is finite $\mu_0$-almost surely, and therefore it makes sense to define the relative partition function
\begin{equation}
z_r:=\int_{\gH} \exp\Big(-\pscal{u\otimes u,w u\otimes u}\Big)\,d\mu_0(u),
\label{eq:def_Z_r}
\end{equation}
which is a number in $(0,1]$. The nonlinear Gibbs measure is then 
\begin{equation}
\boxed{d\mu(u)=z_r^{-1}\exp\Big(-\pscal{u\otimes u,w u\otimes u}\Big)\,d\mu_0(u).}
\label{eq:def_Gibbs}
\end{equation}
Like for $\mu_0$, the corresponding $k$-particle density matrices are in the trace-class since we have the operator inequality
\begin{equation}
\int_{\gH} |u^{\otimes k}\rangle\langle u^{\otimes k}|\,d\mu(u)\leq (z_r)^{-1}\int_{\gH} |u^{\otimes k}\rangle\langle u^{\otimes k}|\,d\mu_0(u)=(z_r)^{-1}k!\,(h^{-1})^{\otimes k}.
\label{eq:estim_DM_0}
\end{equation}

The many-particle interacting quantum system is described by the Hamiltonian
\begin{equation} \label{eq:bH-lambda}
\bH_\lambda= \bH_0 +\lambda \mathbb{W} = 0 \oplus h \oplus \bigoplus_{m=2}^\infty \left( \sum_{i=1}^m h_i + \lambda \sum_{1\le i<j\le m} w_{ij} \right)
\end{equation}
on the Fock space $\cF(\gH)$. Similarly to the non-interacting case, we consider the Gibbs state  
\begin{align} \label{eq:Gibb-int}
\Gamma_{\lambda,T}=Z_{\lambda}(T)^{-1}\,\exp\left(-\frac{\bH_\lambda}{T}\right)\quad\text{with}\quad Z_{\lambda}(T)=\tr_{\cF(\gH)} \left\{\exp\left(-\frac{\bH_\lambda}{T}\right)\right\}.
\end{align}
That $Z_{\lambda}(T)$ is finite follows from the fact that $\bH_\lambda\geq \bH_0$ since $\mathbb{W}\geq0$, which gives
$$Z_{\lambda}(T)\leq\tr_{\cF(\gH)} \left\{\exp\left(-\frac{\bH_0}{T}\right)\right\}=Z_{0}(T)<\ii.$$

Our main theorem states that the density matrices of the interacting quantum state converge to that given by the nonlinear Gibbs measure, provided that the coupling constant scales as
$$\lambda\sim1/T$$
which, as will be explained below, places us in the mean-field regime.

\begin{theorem}[\textbf{Convergence to nonlinear Gibbs measures, $p=1$}]\label{thm:main}\mbox{}\\
Let $h\geq0$ and $w\geq0$ be two self-adjoint operators on $\gH$ and $\gH\otimes_s\gH$ respectively, such that 
\begin{equation}
\tr_{\gH}\big(h^{-1}\big)+\tr_{\gH\otimes_s\gH}\big(w\,h^{-1}\otimes h^{-1}\big)<\ii.
\label{eq:assumptions_h_w}
\end{equation}
Let $\Gamma_{\lambda,T}$ be the grand-canonical quantum Gibbs state defined in \eqref{eq:Gibb-int} and let $Z_\lambda(T)$ be the corresponding partition function. Then we have
\begin{equation}
\boxed{\lim_{\substack{T\to\ii\\ T\lambda\to1}}\frac{Z_\lambda(T)}{Z_0(T)}=z_r}
\label{eq:CV_partition_fn}
\end{equation}
where $z_r$ is the nonlinear relative partition function defined in~\eqref{eq:def_Z_r}.
Furthermore, the nonlinear Gibbs measure $\mu$ defined in~\eqref{eq:def_Gibbs} is the (unique) de Finetti measure of $(\Gamma_{\lambda,T})$ at scale $1/T$: we have the convergence
\begin{equation}
\boxed{\tr_{\cF(\gH)}\big(\Gamma_{\lambda,T}\mathbb{B}_{1/T}\big)\underset{\substack{T\to\ii\\ T\lambda\to1}}{\longrightarrow}\int_{\gH} b(u)\,d\mu(u)}
\label{eq:CV_anti-Wick}
\end{equation}
for every $b\in C^0_b(V)$ with $V$ an arbitrary finite-dimensional subspace of $\gH$ and $\mathbb{B}_{1/T}$ its anti-Wick quantization~\eqref{eq:def_quantization_b}, as well as 
\begin{equation}
\boxed{\frac{k!}{T^k}\Gamma_{\lambda,T}^{(k)} \underset{\substack{T\to\ii\\ T\lambda\to1}}{\longrightarrow}\int_\gH |u^{\otimes k}\rangle\langle u^{\otimes k}|\,d\mu(u)}
\label{eq:CV_DM}
\end{equation}
strongly in the trace-class for every fixed $k\geq1$. 
\end{theorem}

When $\gH$ is a finite-dimensional space, a version of this theorem was proved in~\cite{Gottlieb-05}, in a canonical setting (see also~\cite{GotSch-09,JulGotMarPol-13}). A grand-canonical analogue is treated in~\cite[Chapter 3]{Knowles-thesis}, where also the time-dependent correlation functions of the Gibbs states are considered. We refer to~\cite[Appendix~B]{Rougerie-cdf} for a discussion of the finite dimensional setting more related to that we shall do here. Before turning to the more general case $p>1$, we discuss the scaling of the coupling constant $\lambda$. 

\subsubsection*{Intuitive picture: the semi-classical regime}
For simplicity of notation, let us assume that $\gH=L^2(I)$ with $I$ a bounded interval in $\R$, that $h=-d^2/dx^2$ and that $w$ is the multiplication operator by $W(x-y)$. In physics, the operator $\bH$ is then often written using the creation and annihilation operators $a(x)$ and $a(x)^\dagger$ of a particle at $x\in I$. These are operator-valued distributions which are such that
$$a(f)^\dagger=\int_I f(x)\,a(x)^\dagger\,dx$$
where $a(f)^\dagger$ is the usual creation operator recalled in Section~\ref{sec:Fock2} below. They satisfy the canonical commutation relations
$$a(y)\,a(x)^\dagger-a(x)^\dagger a(y)=\delta(x-y).$$
Then one may write
\begin{equation}\label{eq:sec quant hamil}
\bH_\lambda = \int_I\nabla a(x)^\dagger\cdot\nabla a(x)dx+ \frac\lambda2 \iint_{I\times I} W(x-y) a(x)^\dagger a(y)^\dagger a(y)a(x)\,dx\,dy. 
\end{equation}
The convergence of the non-interacting density matrices in Lemma~\ref{lem:CV_free} can be reformulated by saying that 
\begin{equation}\label{eq:heuristic picture}
\frac{\pscal{a^\dagger(f_1)\cdots a^\dagger(f_k)a(g_k)\cdots a(g_1)}_{\Gamma_{0,T}}}{T^k}
\end{equation}
admits a limit as $T\to\ii$ for every $k\geq1$ and every functions $f_i,g_i$. We expect the same to be true for the interacting Gibbs state. Therefore, we want to think of $1/\sqrt{T}$ as a semi-classical parameter and we introduce new creation and annihilation operators 
$$b:=a/\sqrt{T}, \quad b^\dagger:=a^\dagger/\sqrt{T}.$$
We expect that the operators $b$ and $b ^{\dagger}$ are bounded independently of $T$ when tested against our Gibbs state $\Gamma_{\lambda,T}$, similarly as in~\eqref{eq:heuristic picture}. The choice $\lambda\sim 1/T$ is then natural to make the two terms of the same order in~\eqref{eq:sec quant hamil}. When $\tr(h^{-1})<\ii$ the interacting Gibbs state will be proved to have $O(T)$ particles in average, which then confirms the previous intuition. As we will explain in the next section, the situation is more involved if $\tr(h^{-1})=\ii$.

The commutator between the new operators tends to zero in the limit:
$$b(y)\,b(x)^\dagger-b(x)^\dagger b(y)=\frac{\delta(x-y)}{T}\to0.$$
This will lead to the effective nonlinear classical model. The latter is obtained by replacing $b(x)$ by a function $u(x)$, $b(x)^\dagger$ by its adjoint $\overline{u(x)}$, leading to the nonlinear Hartree energy
$$\cE(u)=\int_I|\nabla u(x)|^2\,dx+\frac12\iint_{I\times I} W(x-y) |u(x)|^2|u(y)|^2\,dx\,dy.$$
The formal semi-classical limit is obtained by replacing the trace by an integral over the phase space (here $\gH=L^2(I)$):
\begin{equation}
\tr_{\cF(\gH)}\left\{\exp\left(-\frac{\bH_{1/T}}T\right)\right\}\underset{T\to\ii}\sim
\left(\frac{T}{\pi}\right)^{\dim\gH}\int_{\gH}\exp\left(-\cE(u)\right)\,du.
\label{eq:semi-classics}
\end{equation}
Usually, $\dim\gH$ is infinite and the previous expression does not make sense. Fortunately, the infinite constant $(T/\pi)^{\dim\gH}$ cancels when we consider the relative partition function $Z_\lambda(T)/Z_0(T)$ and this is how the nonlinear relative partition function $z_r$ arises. 

In a finite dimensional space $\gH$, the limit~\eqref{eq:semi-classics} can be justified by well-known semi-classical techniques, and this provides an alternative proof of the main result in~\cite{Gottlieb-05}, details of which may be found in~\cite[Appendix B]{Rougerie-cdf}. In that case $Z_\lambda(T)$ and $Z_0(T)$ may be studied separately and there is no need to consider the ratio $Z_\lambda(T)/Z_0(T)$. Our goal in the next section will be to adapt these methods to the case of an infinite-dimensional space~$\gH$. The main idea is of course to always deal with relative quantities instead of dangerous expressions like in~\eqref{eq:semi-classics}.

\subsection{The general case}
When 
$$\tr(h^{-1})=+\ii\quad\text{but}\quad  \tr(h^{-p})<\ii\quad \text{for some $p>1$,}$$ 
we are forced to consider the negative Sobolev type space $\gH^{1-p}$, where our final nonlinear Gibbs measure lives. The intuitive picture discussed above is more involved: even if we still expect that the new creation and annihilation operators $b$ and $b^\dagger$ are of order one (when tested against $\Gamma_{\lambda, T}$), the average particle number itself grows faster than $T$,
$$\frac{\tr_{\cF(\gH)}(\cN\Gamma_{\lambda,T})}{T}=\int \langle b^\dagger(x)\,b(x)\rangle_{\Gamma_{\lambda,T}}\,dx\underset{\substack{T\to\ii\\ \lambda T\to 1}}\longrightarrow+\ii,$$
leading to a state with infinitely many particles in the limit. The main difficulty is then the lack of control on the density matrices, which are all unbounded in the trace class.

In order to deal with this pathological case, we will put strong assumptions on the interaction operator $w$ which do not cover anymore functions of the form $W(x-y)$ when $\gH=L^2(\Omega)$. We would like to think of $w$ as a regularized interaction and we do not discuss here how the regularization could be removed by using renormalization techniques.
We therefore assume that $w$ is a bounded non-negative self-adjoint operator on $\gH\otimes_s\gH$ which satisfies an estimate of the form
\begin{equation}
0\leq w\leq  Ch^{1-p'}\otimes h^{1-p'}
\label{eq:assumption_w_p}
\end{equation}
for some $p'> p$. Under this assumption, $w$ naturally extends to an operator from $\gH^{1-p}\otimes_s\gH^{1-p}$ to its predual $\gH^{p-1}\otimes_s\gH^{p-1}$ and we have, similarly as before,
\begin{equation*}
\int_{\gH^{1-p}} \pscal{u\otimes u,w u\otimes u}\,d\mu_0(u)=\tr_{\bigotimes_s^2\gH}\left(w\,h^{-1}\otimes h^{-1}\right)\leq C[\tr(h^{-{p'}})]^2<\ii.
\end{equation*}
In particular, $\pscal{u\otimes u,w u\otimes u}$ (interpreted as a duality product on $\gH^{p-1}\otimes_s\gH^{p-1}$) is finite $\mu_0$ almost surely and this allows us to define the interacting measure $\mu$ on $\gH^{1-p}$ by the same formula~\eqref{eq:def_Gibbs} as in the previous section. The typical example we have in mind is that of $w$ a finite rank operator, with smooth eigenvectors in $\gH^{p'-1}\otimes_s\gH^{p'-1}$.

Our result in the case $p>1$ is similar but weaker than Theorem~\ref{thm:main}. 

\begin{theorem}[\textbf{Convergence to nonlinear Gibbs measures, $p>1$}]\label{thm:main2}\mbox{}\\
Let $h\geq0$ and $w\geq0$ be two self-adjoint operators on $\gH$ and $\gH\otimes_s\gH$ respectively. We assume that 
\begin{itemize}
\item $\tr_{\gH}\big(h^{-p}\big)<\ii$ for some $1<p<\ii$,
\item $0\leq w\leq  Ch^{1-p'}\otimes h^{1-p'}$ for some $p'>p$.
\end{itemize}
Let $\Gamma_{\lambda,T}$ be the grand-canonical quantum Gibbs state defined in \eqref{eq:Gibb-int} and let $Z_\lambda(T)$ be the corresponding partition function. Then $\mu$ is the unique de Finetti measure of $\Gamma_{\lambda,T}$ at scale $1/T$ and we have the same convergence results~\eqref{eq:CV_partition_fn} and~\eqref{eq:CV_anti-Wick} as in Theorem~\ref{thm:main} (with $\gH$ replaced by $\gH^{1-p}$).
Furthermore, 
\begin{equation}
\boxed{\frac{\Gamma_{\lambda,T}^{(1)}}{T} \underset{\substack{T\to\ii\\ T\lambda\to1}}{\longrightarrow}\int_{\gH^{1-p}} |u\rangle\langle u|\,d\mu(u)}
\label{eq:limit_DM_S_p}
\end{equation}
strongly in the Schatten space $\gS^p(\gH)$.
\end{theorem}

The limit for the partition functions and for the anti-Wick observables is the same as in Theorem~\ref{thm:main}, but we are only able to deal with the one-particle density matrix, in the Schatten space $\gS^p$. Unfortunately, although the limiting density matrices
$$\gamma^{(k)}:=\int_{\gH^{1-p}}|u^{\otimes k}\rangle\langle u^{\otimes k}|\,d\mu(u)$$
belong to $\gS^{p}(\bigotimes_s^k\gH)$ for all $k\geq2$ by the same argument as in~\eqref{eq:estim_DM_0}, we are not able to derive the appropriate estimates on $\Gamma_{\lambda,T}^{(k)}$ in $\gS^p$, when $k\geq2$, and obtain the convergence of higher density matrices. We hope to come back to this problem in the future.

\medskip

The rest of the paper is devoted to the proof of our results. Our approach is based on de Finetti measures (which we have already defined before), in the same spirit as in~\cite{AmmNie-08,LewNamRou-14}, and on a Berezin-Lieb type inequality for the relative entropy. These tools are introduced in the next two sections, before we turn to the actual proof of the main results.

%%%%%%%%%%%%%%%%%%%%%%%%%%%%%%%%%%%%%%%%%%
%%%%%%%%%%%%%%%%%%%%%%%%%%%%%%%%%%%%%%%%%%
\section{Construction of de Finetti measures via coherent states} \label{sec:deFinetti_coherent}
%%%%%%%%%%%%%%%%%%%%%%%%%%%%%%%%%%%%%%%%%%
%%%%%%%%%%%%%%%%%%%%%%%%%%%%%%%%%%%%%%%%%%

In this section we provide the proof of Theorem~\ref{thm:deFinetti} and we discuss in details the construction of de Finetti measures.
We follow here ideas from~\cite{AmmNie-08,LewNamRou-14} and we start with a complement to Section~\ref{sec:Fock} on the classical  properties of Fock spaces, which include those of coherent states.

\subsection{Grand canonical ensemble II}\label{sec:Fock2}

We start by defining creation and annihilation operators, which are useful when dealing with coherent states.

\subsubsection*{Creation and annihilation operators}
In the Fock space $\cF(\gH)$, it is useful to introduce operators relating the different $n$-particle sectors. The creation operator $a^\dagger(f)$ on $\cF(\gH)$ is defined for every $f\in \gH$ by
$$a^\dagger(f)\big(\psi_0\oplus\psi_1\oplus\cdots\big)=0\oplus (\psi_0 f)\oplus (f\otimes_s\psi_1)\oplus\cdots.$$
In particular, it maps an $n$-particle state to an $(n+1)$-particle state. Its (formal) adjoint is denoted by $a(f)$ and it is anti-linear with respect to $f$. Its action on a $n$-particle vector is
\begin{multline*}
a(f)\big(g_1\otimes_s\cdots \otimes_s g_n\big)=\pscal{f,g_1}g_2\otimes_s\cdots \otimes_s g_n+\cdots + \pscal{f,g_n}g_1\otimes_s\cdots \otimes_s g_{n-1}.
\end{multline*}
These operators satisfy the canonical commutation relations
$$\begin{cases}
a(g)a^\dagger(f)- a^\dagger(f)a(g)=\pscal{g,f}\1_{\cF(\gH)}\\
a(g)a(f)- a(f)a(g)=0.
  \end{cases}
$$

The creation and annihilation operators are useful to express some relevant physical quantities. For instance, the $k$-particle density matrix of a state $\Gamma$, defined in~\eqref{eq:def_DM_partial}, is characterized by the property that 
\begin{equation}
\langle f_{1}\otimes_s\cdots\otimes_sf_k,\Gamma^{(k)}g_1\otimes_s\cdots \otimes_sg_k\rangle = \tr_{\F(\gH)} \Big( a^\dagger(g_k) ...a^\dagger(g_1) a(f_1)...a(f_k)\Gamma\Big) 
\label{eq:def_DM}
\end{equation}  
for all $f_1,...,f_k,g_1,...,g_k\in \gH$.

\subsubsection*{Localization in Fock space}
If the one-particle Hilbert space $\gH$ is a direct sum of two subspaces, $\gH=\gH_1\oplus\gH_2$, then the corresponding Fock spaces satisfy the factorization property
\begin{equation}
\cF(\gH)\simeq \cF(\gH_1)\otimes\cF(\gH_2)
\label{eq:factorization_Fock_space}
\end{equation}
in the sense of unitary equivalence. For basis sets $(f_i)$ and $(g_i)$ of $\gH_1$ and $\gH_2$ respectively, the unitary map is simply
$$f_1\otimes_s\cdots\otimes_s f_n\otimes_sg_1\otimes_s\cdots \otimes_s g_m\mapsto (f_1\otimes_s\cdots\otimes_s f_n)\otimes (g_1\otimes_s\cdots \otimes_s g_m).$$
The two corresponding annihilation operators on $\cF(\gH_1)$ and $\cF(\gH_2)$ are just 
$a_1(f)\simeq a(P_1f)$ and $a_2(g)\simeq a(P_2g)$
where $1=P_1+P_2$ are the corresponding orthogonal projections. Then, for any state $\Gamma$ on $\cF(\gH)$, we define its \emph{localization $\Gamma_{P_1}$ in $P_1$} (which we also sometimes denote by $\Gamma_{\gH_1}$) as the partial trace
$$\Gamma_{P_1}:=\tr_{\gH_2}[\Gamma]$$
or, equivalently,
$$\tr_{\cF(\gH_1)}[A\Gamma_{P_1}]=\tr_{\cF(\gH)}[A\otimes \1_{\cF(\gH_2)}\Gamma]$$
for every bounded operator $A$ on $\cF(\gH_1)$. The localized state is always a state, that is, it satisfies $\Gamma_{P_1}\geq0$ and $\tr_{\cF(\gH_1)}\Gamma_{P_1}=1$.
For a state commuting with $\cN$ and having finite moments to any order\footnote{I.e. $\tr_{\cF (\gH)} [\cN ^\alpha \Gamma] < C_{\alpha}$ for any $\alpha>0$.}, $\Gamma_P$ is characterized by the property that the density matrices are localized in a usual sense:
$$(\Gamma_P)^{(k)}=P^{\otimes k}\Gamma^{(k)}P^{\otimes k},\qquad \forall k\geq1.$$
This can be used to provide an explicit formula for $\Gamma_P$, see~\cite[Rmk 13 \& Ex. 10]{Lewin-11}. 
Localization is a fundamental concept for many-particle quantum systems. See for instance~\cite[Appendix A]{HaiLewSol_2-09} for its link with entropy and~\cite{Lewin-11} for the associated weak topology.

\subsubsection*{Coherent states}
For every vector $u$ in $\gH$ (which is not necessarily normalized), we denote the Weyl unitary  operator on  $\cF(\gH)$ by $$W(u):=\exp(a^\dagger(u)-a(u)).$$
A coherent state is a Weyl-rotation of the vacuum:
$$
\xi(u):=W(u) |0\rangle := \exp(a^\dagger(u)-a(u))  |0\rangle = e^{-|u|^2/2} \bigoplus_{n \ge 0} \frac{1}{\sqrt{n!}} u^{\otimes n} .
$$
We have already mentioned the resolution of the identity on any finite-dimensional subspace $V\subset\gH$:
\begin{equation}
\int_{V} W(u)|0\rangle  \langle 0| W(u)^* du = \left( \int_{V} e^{-|u|^2} du \right) \1_{\F(V)} = \pi^{\dim(V)}  \1_{\F(V)}.
\label{eq:resolution_coherent2}
\end{equation}
With the identification $\cF(\gH)=\cF(V)\otimes\cF(V^\perp)$, $\cF(V)$ is identified to $\cF(V)\otimes|0_{V^\perp}\rangle$ where $|0_{V^\perp}\rangle$ is the vacuum of $\cF(V^\perp)$.
Coherent states satisfy many interesting algebraic properties. For instance, creation and annihilation operators are translated by a constant when rotated by a Weyl unitary:
\begin{equation}
 \label{eq:Weyl-action}
W(f)^* a^\dagger(g) W(f)=a^\dagger(g) + \langle f,g \rangle, \quad W(f)^* a(g) W(f)=a(g) + \langle g,f \rangle.
\end{equation}
From this we conclude that coherent states are eigenfunctions of the annihilation operator:
\begin{equation}
a(g) \xi(f)=\langle g,f \rangle \xi(f),\qquad\forall f,g\in\gH.
\label{eq:coherent_eigenfn}
\end{equation}
Similarly, the $k$-particle density matrix of $\xi(u)$ is
\begin{equation}
\big[|\xi(u)\>\<\xi(u)|\big]^{(k)}=|u^{\otimes k}\>\<u^{\otimes k}|.
\label{eq:DM_coherent}
\end{equation}

\subsection{Finite-dimensional lower symbols}\label{sec:Husimi}
For any state $\Gamma$ on $\cF(\gH)$ and any finite-dimensional subspace $V\subset\gH$, we define the \emph{lower symbol} (or Husimi function) on $V$ by
\begin{equation} \label{eq:Husimi}
\mu_{V,\Gamma}^{\eps}(u):=(\eps\pi)^{-\dim(V)}\pscal{\xi(u/\sqrt{\eps}),\Gamma_V\xi(u/\sqrt{\eps})}_{\cF(\gH)},
\end{equation}
where $\Gamma_V=\tr_{\cF(V^\perp)}(\Gamma)$ is the associated localized state defined using the isomorphism of Fock spaces $\cF(\gH)\simeq\cF(V)\otimes \cF(V^\perp)$ recalled above.

\begin{lemma}[\textbf{Cylindrical projections}]\label{lem:cylindrical}\mbox{}\\
The cylindrical projection of $\mu^\eps_{V_1,\Gamma}$ onto a subspace $V_2\subset V_1$ is $\mu^\eps_{V_2,\Gamma}$.
\end{lemma}
\begin{proof}
We denote $u=v_2+v_2^\perp$ with $v_2\in V_2$ and $v_2^\perp\in V_2^\perp=V_1\ominus V_2$. Then we remark that
$W(u)=W(v_2)W(v_2^\perp)$
due to the fact that $a(v_2)$ and $a(v_2^\perp)$ commute and, therefore,
$$|\xi(u)\>\<\xi(u)|\simeq |\xi(v_2)\>\<\xi(v_2)|\otimes |\xi(v_2^\perp)\>\<\xi(v_2^\perp)|.$$
In particular,
$$(\pi)^{-\dim(V_1)}\int_{V_2^\perp}|\xi(v_2+v_2^\perp)\>\<\xi(v_2+v_2^\perp)|\, dv_2^\perp= (\eps\pi)^{-\dim(V_2)}|\xi(v_2)\>\<\xi(v_2)|$$
and the result follows.
\end{proof}

\subsubsection*{Link with density matrices}
With the Husimi function $\mu^\eps_{V,\Gamma}$ at hand, it is natural to consider the state in the Fock space $\cF(V)$
$$\int_{V}|\xi(u/\sqrt{\eps})\>\<\xi(u/\sqrt{\eps})|\;d\mu^\eps_{V,\Gamma}(u)$$
and its $k$-particle density matrix
$$k! \eps^{-k}\int_{V}|u^{\otimes k}\>\<u^{\otimes k}|\;d\mu^\eps_{V,\Gamma}(u).$$
A natural question is to ask whether these density matrices approximate the density matrices $\Gamma^{(k)}$ of the original state $\Gamma$, in the limit $\eps\to0$. We have the following explicit formula.

\begin{lemma}[\textbf{Lower symbols and density matrices}]\label{lem:link_PDM_Husimi}\mbox{}\\
We have on $\bigotimes_s^k V$
\begin{equation}
\boxed{\int_{V}|u^{\otimes k}\>\<u^{\otimes k}|\;d\mu^\eps_{V,\Gamma}(u)=k!\eps^k\sum_{\ell=0}^{k} {k\choose \ell}\Gamma_V^{(\ell)} \otimes _s \1_{\otimes^{k-\ell}_sV}}
\label{eq:relate_PDM_Husimi}
\end{equation}
where we recall the convention~\eqref{eq:A_k} and where $\Gamma_V^{(\ell)}=(P_V)^{\otimes \ell}\Gamma^{(\ell)}(P_V)^{\otimes \ell}$ is the $\ell$-particle density matrix of the localized state $\Gamma_V$ on $V$.
In particular, we have the upper bound
\begin{equation}
{k!}\eps^k\,(P_V)^{\otimes k}\Gamma^{(k)}(P_V)^{\otimes k}\leq \int_{V}|u^{\otimes k}\>\<u^{\otimes k}|\;d\mu^\eps_{V,\Gamma}(u)
\label{eq:estim_PDM_Husimi}
\end{equation}
as operators on $\bigotimes_s^kV$. 
\end{lemma}

The lemma is similar to~ \cite[eq. (6)]{Chiribella-11} and \cite[Thm.~2.2]{LewNamRou-14b}, see also~\cite{Harrow-13}.

\begin{proof} By~\cite[Lem.~4.1]{LewNamRou-14b}, expectations against Hartree products determine the state and it therefore suffices to prove that  
$$\int_{V}|\pscal{u,v}|^{2k}\;d\mu^\eps_{V,\Gamma}(u)={k!}\eps^k\sum_{\ell=0}^{k} {k\choose\ell}\pscal{v^{\otimes \ell},\Gamma_V^{(\ell)}v^{\otimes \ell}}.$$ 
From the definition of the $\ell$-particle density matrix, we have 
$$\pscal{v^{\otimes \ell},\Gamma_V^{(\ell)}v^{\otimes \ell}}=\frac{\pscal{v^{\otimes_s \ell},\Gamma_V^{(\ell)}v^{\otimes_s \ell}}}{\ell!}=\frac{\tr\big[a^\dagger(v)^\ell a(v)^\ell\Gamma\big]}{\ell!}.$$
On the other hand, using~\eqref{eq:coherent_eigenfn} we can write 
\begin{align*}
&\eps^{-k}\int_{V}|\pscal{u,v}|^{2k}\;d\mu^\eps_{V,\Gamma}(u)\\
&\qquad\qquad=(\eps\pi)^{-\dim(V)}\eps^{-k}\int_{V}|\pscal{u,v}|^{2k}\pscal{\xi(u/\sqrt{\eps}),\Gamma\xi(u/\sqrt{\eps})}\,du\\
&\qquad\qquad=(\eps\pi)^{-\dim(V)}\int_{V}\pscal{a(v)^k\xi(u/\sqrt{\eps}),\Gamma a(v)^k\xi(u/\sqrt{\eps})}\,du\\
&\qquad\qquad=\tr \big[a(v)^ka^\dagger(v)^k\Gamma\big] .
\end{align*}
Therefore the result follows from the formula
\begin{equation}\label{eq:Wick A Wick}
a(v) ^k a ^\dagger(v) ^k = \sum_{\ell=0} ^k {k\choose\ell} \frac{k!}{\ell!} a ^\dagger(v) ^\ell a (v) ^\ell,
\end{equation}
see for instance~\cite[Lem.~4.2]{LewNamRou-14b}.
\end{proof}

\begin{remark}[Moments estimates]\mbox{}\\
It is clear from formula~\eqref{eq:relate_PDM_Husimi} that in any finite-dimensional space $V$, we have
\begin{equation}
\int_V \|u\|^{2s}d\mu_{V,\Gamma}^{\eps}(u)\leq C_{s,V}\tr_{\cF(V)}\big[(\eps\cN)^s\Gamma_V\big]
\label{eq:moment_estimates}
\end{equation}
for any state $\Gamma$ and all integers $s\geq1$ (the constant $C_{s,V}$ only depends on $s$ and $\dim(V)$). By complex interpolation, we deduce that the same holds for any real number $s\geq1$.
\end{remark}

\begin{remark}[A quantitative bound]\mbox{}\\
In the same spirit as~\cite{Chiribella-11,LewNamRou-14b}, from \eqref{eq:relate_PDM_Husimi} and~\eqref{eq:trace_A_k} we deduce the following quantitative estimate 
\begin{align}
&\norm{k!\eps^k\Gamma^{(k)}_V-\int_{V}|u^{\otimes k}\>\<u^{\otimes k}|\;d\mu^\eps_{V,\Gamma}(u)}_{\gS^1(\otimes_s^kV)}\nn\\
& \qquad\qquad=k!\eps^k\sum_{\ell=0}^{k-1}{k\choose \ell} \tr \Big( \Gamma_V^{(\ell)} \otimes _s \1_{\otimes^{k-\ell}_sV} \Big)   \nn\\
&\qquad\qquad\leq k!\eps^k\sum_{\ell=0}^{k-1}{k\choose \ell} {{k-\ell+d-1}\choose{d-1}} \tr\big[\Gamma_V^{(\ell)}\big]\nn\\
&\qquad\qquad\leq \eps^k\,\sum_{\ell=0}^{k-1}{k\choose \ell}^2\frac{(k-\ell+d-1)!}{(d-1)!}\tr\big[\cN^{\ell}\Gamma_V\big]
\label{eq:quantitative}
\end{align}
with $d=\dim(V)$.
We remark that if $\Gamma$ is an $N$-particle state and $\eps=1/N$, then the error term is
$$\frac1N\sum_{j=0}^{k-1}{k\choose {j+1}}^2\frac{(j+d)!}{(d-1)!}\frac{1}{N^{j}}.$$
For fixed $k,d\geq1$ the latter behaves as $k^2d/N$ in the limit $N\to\ii$. This is similar to the bound $4kd/N$ proved in~\cite{ChrKonMitRen-07} and reviewed in~\cite{LewNamRou-14b}, which was based on a resolution of the identity involving Hartree states instead of coherent states.
\end{remark}

We are now able to provide the proof of Theorem~\ref{thm:deFinetti}.

%%%%%%%%%%%%%%%%%%%%%%%%%%%%%%%%%%%%%%%%%%
\subsection{Construction of de Finetti measures: proof of Theorem~\ref{thm:deFinetti}}\label{sec:proof_de_Finetti}

We have 
\begin{equation}
\tr\big((\epsilon_n h^{1-p})^{\otimes k}\Gamma_n^{(k)}\big)\leq C_k
\label{eq:bound_DM_proof}
\end{equation}
for all $1\leq k\leq \kappa$ (all $1\leq k<\kappa$ if $\kappa=+\ii$), due to our assumption~\eqref{eq:bound_DM}. 
Let $\{u_j\}$ be a basis of eigenvectors of $h$. First we fix the finite-dimensional space $V=V_J:={\rm span}(u_1,...,u_J)$ with corresponding orthogonal projection $P=P_J$. Let $\{\Gamma_n\}_P$ be the $P$-localized state in the Fock space $\cF(V)$, whose density matrices are $(\Gamma_n)_P^{(k)}=P^{\otimes k}\Gamma_n^{(k)}P^{\otimes k}$ and we denote by $\mu_{V,\Gamma_n}^{\eps_n}$ the Husimi function as defined in Section~\ref{sec:Husimi}. 
Inserting Formula~\eqref{eq:relate_PDM_Husimi} and arguing as in~\eqref{eq:quantitative}, we find that
\begin{equation}
\Big|\int_V \norm{u}^{2k}_{\gH^{1-p}}d\mu_{V,\Gamma_n}^{\eps_n}(u)-k!\eps_n^k\tr\big[(h^{1-p})^{\otimes k}(P_V)^{\otimes k}\Gamma_V^{(k)}(P_V)^{\otimes k}\big]\Big|\leq C_{k,V}\eps_n
\label{eq:test_h_p}
\end{equation}
with a constant that depends on $V$ and on the constants in ~\eqref{eq:bound_DM_proof}. In particular, we have
\begin{equation}
\int_V \norm{u}^{2k}_{\gH^{1-p}}d\mu_{V,\Gamma_n}^{\eps_n}(u)\leq C_k+C_{k,V}\eps_n,\qquad \forall 1\leq k\leq \kappa.
\label{eq:estim_moments_mu_Husimi}
\end{equation}
This is for integer $k$ but, as in~\eqref{eq:moment_estimates}, we have a similar estimate
\begin{equation}
\int_{V}\|u\|_{\gH^{1-p}}^{2s}d\mu_{V,\Gamma_n}^{\eps_n}(u)\leq C_{s,V}
\label{eq:tight_moments}
\end{equation}
for all $1\leq s\leq\kappa$ ($1\leq s<\kappa$ if $\kappa=+\ii$). Therefore the sequence $\mu_{V,\Gamma_n}^{\eps_n}$ is tight and we may assume, after extraction of a subsequence, that $\mu_{V,\Gamma_n}^{\eps_n}$ converges to a probability measure $\nu_{V}$ on $V$. From the moment estimates~\eqref{eq:tight_moments}, we also have 
\begin{equation}
\int_V |u^{\otimes k}\>\<u^{\otimes k}|d\mu_{V,\Gamma_n}^{\eps_n}(u)\to\int_V |u^{\otimes k}\>\<u^{\otimes k}|d\nu_{V}(u)
\label{eq:CV_PDM_measure}
\end{equation}
(the convergence is in a finite-dimensional space, hence strong). By a diagonal argument, we can assume convergence to a measure $\nu_{V_J}$ for every $J\geq1$.

From the bound~\eqref{eq:estim_moments_mu_Husimi} we get
$$\int_V \norm{u}^{2k}_{\gH^{1-p}}d\nu_{V}(u)\leq C_k$$
with a constant depending on $k$ but not on $V$. Using this for $k=1$, we deduce immediately from~\cite[Lemma 1]{Skorokhod-74} that there exists a measure $\nu$ on $\gH^{1-p}$ whose cylindrical projections are $\nu_V$, since 
\begin{equation*}
R^2\nu_V\Big(\{u\in V\ :\ \norm{u}_{\gH^{1-p}}\ge R\}\Big)\leq \int_V\norm{u}_{\gH^{1-p}}^2\,d\nu_V(u)\leq C.
\label{eq:tightness_verification}
\end{equation*}

Finally, recalling~\eqref{eq:def_quantization_b}, we can write for every $b\in C^0_b(V_J)$
$$\int_{V_J}b(u)\,d\mu_{V,\Gamma_n}^{\eps_n}=\tr\big[\mathbb{B}_{\eps_n}\Gamma_n]$$
and therefore the convergence against anti-Wick observables as in~\eqref{eq:limit_Anti-Wick} follows immediately from the convergence of the Husimi measures,  if $V=V_J$ for some $J$. As we have already mentioned in Section~\ref{sec:deFinetti}, the case of a general $V$ follows from~\eqref{eq:unif_bound_Anti_Wick} by density.

As for~\eqref{eq:weak_limit_DM}, we note that the bound~\eqref{eq:bound_DM} implies that, after possibly a further extraction,
$$k!(\epsilon_{n_j})^k\<\Psi,\Gamma_{n_j}^{(k)}\Psi\>$$
has a limit for all $1\leq k < \kappa$ and $\Psi\in \otimes_s^k\gH^{p-1}$. The identification of the limit as given by the right-hand side of~\eqref{eq:weak CV deF} proceeds by arguments similar as above.

\qed

%%%%%%%%%%%%%%%%%%%%%%%%%%%%%%%%%%%%%%%%%%
%%%%%%%%%%%%%%%%%%%%%%%%%%%%%%%%%%%%%%%%%%
\section{A lower bound on the relative entropy}\label{sec:entropy}
%%%%%%%%%%%%%%%%%%%%%%%%%%%%%%%%%%%%%%%%%%
%%%%%%%%%%%%%%%%%%%%%%%%%%%%%%%%%%%%%%%%%%

We recall that the relative entropy of two states $\Gamma$ and $\Gamma'$ on $\cF(\gH)$ is defined by
$$\boxed{\cH(\Gamma,\Gamma')=\tr_{\cF(\gH)}\left[\Gamma(\log\Gamma-\log\Gamma')\right].}$$
It is a positive number which can in principle be equal to $+\ii$. An important property of the relative entropy is its monotonicity under 
two-positive trace-preserving maps. That is, if we have a linear map $\Phi:\cB(\gK_1)\to\cB(\gK_2)$ which satisfies 
$$\begin{pmatrix}
A_{11} & A_{12}\\
A_{21} & A_{22}\\
\end{pmatrix}\geq0 \Longrightarrow \begin{pmatrix}
\Phi(A_{11}) & \Phi(A_{12})\\
\Phi(A_{21}) & \Phi(A_{22})\\
\end{pmatrix}\geq0$$ 
for all operators $A_{ij}$ on $\gK_1$ and $\tr(\Phi(A))=\tr(A)$ for all $A$, then 
$$\cH(\Phi(A),\Phi(B))\leq \cH(A,B),$$ 
see~\cite{OhyPet-93,Petz-03}. 

An important two-positive trace-preserving map is the localization $\Gamma\mapsto\Gamma_{V}$ to a subspace $V\subset \gH$, which is defined by using the isomorphism $\cF(\gH)\simeq \cF(V)\otimes\cF(V^\perp)$ and by taking the partial trace with respect to the second Fock space, as recalled in Section~\ref{sec:Fock2}. In particular, we deduce that
\begin{equation}
\cH(\Gamma,\Gamma')\geq \cH\big(\Gamma_{V},\Gamma'_V\big).
\label{eq:monotone}
\end{equation}
It can be shown that 
$$\cH(\Gamma,\Gamma') = \lim_{k\to\ii}\cH\big(\Gamma_{V_k},\Gamma'_{V_k}\big)$$
if $V_k$ is any increasing sequence of subspaces of $\gH$ such that the corresponding orthogonal projections $P_k\to1$ strongly (see~\cite[Cor. 5.12]{OhyPet-93}, as well as~\cite[Thm. 2]{LewSab-14} for a similar result).

In a similar fashion, the relative entropy of two probability measures on a Hilbert space $\gK$ is defined by
$$\boxed{\cHcl(\mu,\mu'):=\int_{\gK}\frac{d\mu}{d\mu'}(u)\log\left(\frac{d\mu}{d\mu'}(u)\right)\,d\mu'(u)}$$
where  $d\mu/d\mu'$ is the density of $\mu$ relatively to $\mu'$. If $\mu$ is not absolutely continuous with respect to $\mu'$, then $\cHcl(\mu,\mu')=+\ii$. The classical relative entropy is also monotone in the sense that 
\begin{equation}
\cHcl(\mu,\mu')\geq \cHcl(\mu_{V},\mu'_{V})
\label{eq:monotone_cl}
\end{equation}
for any subspace $V\subset\gK$, where $\mu_{V}$ and $\mu_{V}'$ are the associated cylindrical projections of $\mu$ and $\mu'$. Similarly as in the quantum case, one has
\begin{equation}
\cHcl(\mu,\mu')= \lim_{k\to\ii}\cHcl(\mu_{V_k},\mu'_{V_k}).
\label{eq:CV_classical_V_k}
\end{equation}

Our next result gives a lower bound on the relative entropy of two quantum states in terms of the corresponding de Finetti measures constructed in Theorem~\ref{thm:deFinetti}.

\begin{theorem}[\textbf{Relative entropy: quantum to classical}] \label{thm:rel-entropy}\mbox{}\\
Let $\Gamma$ and $\Gamma'$ be two states on $\cF(\gH)$ and let $V\subset \gH$ be a finite dimensional-subspace. Then we have
\begin{equation}
\boxed{\cH(\Gamma,\Gamma')\geq \cH(\Gamma_V,\Gamma'_V)\geq \cHcl(\mu_{V,\Gamma}^\eps,\mu_{V,\Gamma'}^\eps)}
\label{eq:Berezin-Lieb}
\end{equation}
where $\mu_{V,\Gamma}^\eps$ and $\mu_{V,\Gamma'}^\eps$ are the Husimi measures defined in Section~\ref{sec:Husimi}.

In particular, if we have $\epsilon_n\to0$ and two sequences of states $\{\Gamma_n\}$ and $\{\Gamma_n'\}$ satisfying the assumptions of Theorem~\ref{thm:deFinetti}, then
\begin{equation}
\boxed{\liminf_{n\to \infty} \cH (\Gamma_n,\Gamma_n') \ge \cH_{\rm cl} (\mu,\mu')}
\label{eq:lower_bound_relative_entropy}
\end{equation}
where $\mu$ and $\mu'$ are the de Finetti probability measures of $\{\Gamma_n\}$ and $(\Gamma'_n)$ respectively, on $\gH^{1-p}$ (after extraction of a subsequence).
\end{theorem}

\begin{proof}
We have already explained that $\cH(\Gamma,\Gamma')\geq \cH(\Gamma_V,\Gamma'_V)$, since localization is a 2-positive trace-preserving map. To prove the second inequality in \eqref{eq:Berezin-Lieb}, we can work in the finite-dimensional space $V$. We are going to use a Berezin-Lieb type inequality for the relative entropy, which is the equivalent of well known techniques for the entropy~\cite{Berezin-72,Lieb-73b,Simon-80}. 

\begin{lemma}[\textbf{Berezin-Lieb inequality for the relative entropy}]\label{lem:Berezin-Lieb}\mbox{}\\
Assume that we have a resolution of the identity in a Hilbert space $\gK$, of the form:
\begin{equation}
\int_{M} |n_x\rangle\langle n_x|\, d\zeta(x)=1_{\gK},
\label{eq:resolution_identity}
\end{equation}
where $\zeta$ is a positive Borel measure on $M\subset \R^N$ and $x\in M\mapsto n_x\in S\gK = \left\lbrace u\in \gK, \norm{u} = 1 \right\rbrace$ is continuous.
For a trace-class operator $A\geq0$ on $\gK$, we define the corresponding Husimi Borel measure (or lower symbol) on $M$ by $$dm_A(x)=\pscal{x,Ax}\,d\zeta(x).$$
Then we have 
\begin{equation}
\boxed{\cH(A,B)\geq \cHcl(m_A,m_B)}
\end{equation}
for any states $A,B\geq0$ on $\gK$.
\end{lemma}

We postpone the proof of this result and apply it directly to our situation. In the Fock space $\cF(V)$ we have the resolution of the identity~\eqref{eq:resolution_coherent} and the Husimi function $\mu_{\Gamma,V}^\eps$ is exactly defined as in the lemma. Therefore we immediately conclude that
$\cH(\Gamma_V,\Gamma'_V)\geq \cHcl(\mu_{V,\Gamma}^\eps,\mu_{V,\Gamma'}^\eps).$

In order to prove~\eqref{eq:lower_bound_relative_entropy}, we first localize to the finite dimensional space $V\subset\gH\subset \gH^{1-p}$ spanned by the $J$ first eigenfunctions of $h$. By monotonicity of the relative entropy and~\eqref{eq:Berezin-Lieb}, we find that 
$$\liminf_{n\to \infty} \cH (\Gamma_n,\Gamma_n')\geq \liminf_{n\to \infty} \cH \big(\{\Gamma_n\}_V,\{\Gamma_n'\}_V\big) \geq \liminf_{n\to\ii}\cHcl(\mu_{V,\Gamma_n}^{\eps_n},\mu_{V,\Gamma'_n}^{\eps_n}).$$
By definition of the de Finetti measures in the proof of Theorem~\ref{thm:deFinetti}, we have $\mu_{V,\Gamma_n}^{\eps_n}\wto\mu_V$ and $\mu_{V,\Gamma'_n}^{\eps_n}\wto\mu'_V$. The relative entropy is jointly convex, hence lower semi-continuous and we get
$$\liminf_{n\to \infty} \cH (\Gamma_n,\Gamma_n')\geq \cHcl(\mu_{V},\mu'_{V}).$$
To conclude \eqref{eq:lower_bound_relative_entropy} we remove the localization by passing to the limit $J\to\ii$ using~\eqref{eq:CV_classical_V_k}. 
\end{proof}

It remains to provide the
\begin{proof}[Proof of Lemma~\ref{lem:Berezin-Lieb}]
We write the proof for a discrete resolution of the identity in a finite-dimensional space $\gK$:
$$\sum_{k=1}^Km_k |x_k\rangle\langle x_k|=1,\qquad m_k\geq0.$$
The general case can be proved by passing to the limit. 
The map
$$A\mapsto \begin{pmatrix}
m_1\,\pscal{x_1,Ax_1}& \\
 & m_2\,\pscal{x_2,Ax_2} & \\
 & & \ddots\\
 & & & m_K\,\pscal{x_K,Ax_K}
\end{pmatrix}$$
is two-positive~\cite{Stinespring-55} and trace-preserving, from $\cB(\gK)$ to $\cB(\C^K)$. So we get immediately from the monotonicity that
\begin{align*}
\cH(A,B) &\geq \sum_{k=1}^K\cH(m_k\pscal{x_k,Ax_k},m_k\pscal{x_k,Bx_k})\\
&=\sum_{k=1}^Km_k\pscal{x_k,Ax_k}\log\left(\frac{\pscal{x_k,Ax_k}}{\pscal{x_k,Bx_k}}\right),
\end{align*}
which is what we wanted.
\end{proof}

In this section we have used that the relative entropy is both monotone with respect to two-positive trace-preserving maps, and jointly convex. These are two (equivalent) properties which play an important role in statistical physics and quantum information theory~\cite{Wehrl-78,OhyPet-93,Carlen-10}. They are indeed also equivalent to the \emph{strong subadditivity of the entropy}, which was proved by Lieb and Ruskai in ~\cite{LieRus-73a,LieRus-73b}.

%%%%%%%%%%%%%%%%%%%%%%%%%%%%%%%%%%%%%%%%%%
%%%%%%%%%%%%%%%%%%%%%%%%%%%%%%%%%%%%%%%%%%
\section{Derivation of the nonlinear Gibbs measures: proofs}\label{sec:proof}
%%%%%%%%%%%%%%%%%%%%%%%%%%%%%%%%%%%%%%%%%%
%%%%%%%%%%%%%%%%%%%%%%%%%%%%%%%%%%%%%%%%%%

In this final section we provide the proof of Theorems~\ref{thm:main} and~\ref{thm:main2}. Some of our arguments are common to the two results and they will be written in the general case $p\geq1$. Some other parts are much simper in the case $p=1$ and will then be singled out.
We will prove that 
\begin{equation} 
\lim_{\substack{T\to \infty\\ \lambda T\to1}} \frac{Z_\lambda(T)}{Z_0(T)} = z_r=\int_{\gH^{1-p}}\exp\left(-\pscal{u^{\otimes 2},wu^{\otimes 2}}\right)\,d\mu_0(u)
\end{equation}
by showing 
\begin{equation} 
-\lim_{\substack{T\to \infty\\ \lambda T\to1}}\log\left(\frac{Z_\lambda(T)}{Z_0(T)}\right) = -\log(z_r).
\label{eq:to_be_shown}
\end{equation}
We recall Gibbs' variational principle which states that 
\begin{align*}
-T\log(\tr(e^{-A/T}))&=\tr(A\Gamma_A)+T\tr\Gamma_A\log(\Gamma_A)\\
&=\min_{\substack{\Gamma\geq0\\ \tr\Gamma=1}}\Big\{\tr(A\Gamma)+T\tr\Gamma\log(\Gamma)\Big\}
\end{align*}
with $\Gamma_A:=e^{-A/T}/\tr(e^{-A/T})$. Also, we recall that 
\begin{align}
-\log\frac{Z_\lambda(T)}{Z_0(T)}&=\cH(\Gamma_{\lambda,T},\Gamma_{0,T})+\frac\lambda{T}\tr \big(w\Gamma_{\lambda,T}^{(2)}\big)\label{eq:formulation_Gibbs}\\
&=\min_{\substack{\Gamma\geq0\\ \tr\Gamma=1}}\left\{\cH(\Gamma,\Gamma_{0,T})+\frac\lambda{T}\tr \big(w\Gamma^{(2)}\big)\right\}\label{eq:min_Gibbs},%\\
%&=-\cH(\Gamma_{0,T},\Gamma_{\lambda,T})+\frac\lambda{T}\tr \big(w\Gamma_{0,T}^{(2)}\big).
\end{align}
where $\cH$ is the relative entropy. In a similar manner, we have
\begin{equation}
-\log(z_r)=\min_{\substack{\nu\ \text{proba.}\\ \text{measure on $\gH^{1-p}$}}}\left\{\cHcl(\nu,\mu_0)+\frac{1}{2}\int_{\gH^{1-p}}\pscal{u\otimes u,wu\otimes u}d\nu(u)\right\}
\label{eq:min_Gibbs_nonlinear}
\end{equation}
with $\mu$ being the unique minimizer.
Our proof goes as follows:

\smallskip

\noindent \textbf{Step 1:} We derive some estimates on the one-particle density matrix of the Gibbs state $\Gamma_{\lambda,T}$, which allow to define the limiting de Finetti measure $\nu$ (after extraction of a subsequence), via Theorem~\ref{thm:deFinetti}.

\smallskip

\noindent \textbf{Step 2:} We prove the upper bound 
\begin{equation}
\limsup_{\substack{T\to \infty\\ \lambda T\to1}}-\log\left(\frac{Z_\lambda(T)}{Z_0(T)}\right) \leq -\log(z_r)
\end{equation}
by using a suitable trial state and finite-dimensional semi-classical analysis.

\smallskip

\noindent \textbf{Steps 3--4:} We show that 
\begin{equation}
\liminf_{\substack{T\to \infty\\ \lambda T\to1}}\cH(\Gamma_{\lambda,T},\Gamma_{0,T})\geq \cHcl(\nu,\mu_0)
\label{eq:to_be_shown_1}
\end{equation}
using Theorem~\ref{thm:rel-entropy} and that 
\begin{equation}
\liminf_{\substack{T\to \infty\\ \lambda T\to1}}\frac{\tr_{\cF(\gH)}(\Gamma_{\lambda,T}^{(2)}w)}{T^2}\geq \frac{1}{2}\int_{\gH^{1-p}}\pscal{u\otimes u,wu\otimes u}d\nu(u)
\label{eq:to_be_shown_2}
\end{equation}
using the definition of the de Finetti measure. This is the more difficult step when $p>1$, because of the lack of control of $T^{-2}\Gamma_{\lambda,T}^{(2)}$. The lower estimates~\eqref{eq:to_be_shown_1} and~\eqref{eq:to_be_shown_2} together with the variational principle provide the lower bound
\begin{equation}
\liminf_{\substack{T\to \infty\\ \lambda T\to1}}-\log\left(\frac{Z_\lambda(T)}{Z_0(T)}\right) \geq -\log(z_r),
\end{equation}
and the equality $\nu=\mu$.

\smallskip

\noindent \textbf{Steps 5--6:} We discuss the strong convergence of the density matrices.

\medskip

Now we provide the details of the proofs.

%%%%%%%%%%%%%%%%%%%%%%%%%%%%%%
\subsubsection*{\bf Step 1: some uniform bounds}
Before starting the proof~\eqref{eq:to_be_shown}, we establish some uniform bounds on $Z_\lambda(T)/Z_0(T)$ and on the one-particle density matrix $\Gamma_{\lambda,T}^{(1)}$, that will be useful throughout the proof. 

\begin{lemma}[\textbf{A priori bounds}]\label{lem:unif_bounds}\mbox{}\\
Let $h>0$ and $w\geq0$ be two self-adjoint operators on $\gH$ and $\gH\otimes_s\gH$ respectively, such that 
$$\tr_{\gH}(e^{-h/T})+\tr_{\gH\otimes_s\gH}(w\,h^{-1}\otimes h^{-1})<\ii.$$ 
Let $\Gamma_{\lambda,T}$ be the grand-canonical quantum Gibbs state defined in \eqref{eq:Gibb-int} and let $Z_\lambda(T)$ be the corresponding partition function. Then we have
\begin{equation}
e^{-(\lambda T)\tr_{\gH\otimes_s\gH} \big(w h^{-1}\otimes h^{-1}\big)} \leq \frac{Z_\lambda(T)}{Z_0(T)}\leq 1
\label{eq:unif_bd2}
\end{equation}
and
\begin{equation}
\tr_{\gH\otimes_s\gH}\big(w\Gamma_{\lambda,T}^{(2)}\big)\leq T^2\tr_{\gH\otimes_s\gH} \big(w\, h^{-1}\otimes h^{-1}\big).
\label{eq:estim_interaction}
\end{equation}
\end{lemma}

\begin{proof}[Proof of Lemma~\ref{lem:unif_bounds}]
First we remark that 
$$Z_\lambda(T)=\tr_{\cF(\gH)}e^{-\frac{\bH+\lambda\mathbb{W}}{T}}\leq \tr_{\cF(\gH)}e^{-\frac{\bH}{T}}=Z_0(T).$$
since $w\geq0$ by assumption. On the other hand, we have by~\eqref{eq:min_Gibbs} with $\Gamma=\Gamma_{0,T}$,
\begin{equation*}
-\log\frac{Z_\lambda(T)}{Z_0(T)}\leq \frac\lambda{T}\tr_{\gH\otimes_s\gH} \big(w\Gamma_{0,T}^{(2)}\big)\leq (\lambda T)\tr_{\gH\otimes_s\gH} \big(w h^{-1}\otimes h^{-1}\big).
\end{equation*}
Here we have used that 
$$\Gamma_{0,T}^{(2)}=\left(\frac{1}{e^{h/T}-1}\right)^{\otimes 2}\leq T^2 h^{-1}\otimes h^{-1}$$
by Lemma~\ref{lem:quasi-free}. Hence we have proved that
\begin{equation}
0\leq -\log\frac{Z_\lambda(T)}{Z_0(T)}\leq (\lambda T)\tr_{\gH\otimes_s\gH} \big(w h^{-1}\otimes h^{-1}\big)
\label{eq:unif_bd}
\end{equation}
which may as well be rewritten as in~\eqref{eq:unif_bd2}. The bound~\eqref{eq:estim_interaction} follows from~\eqref{eq:unif_bd} and the fact that
\begin{align*}
-\log\frac{Z_\lambda(T)}{Z_0(T)}&=\cH(\Gamma_{\lambda,T},\Gamma_{0,T})+(\lambda/T)\tr_{\gH\otimes_s\gH} \big(w\Gamma_{\lambda,T}^{(2)}\big)\\
&\geq (\lambda/T)\tr_{\gH\otimes_s\gH} \big(w\Gamma_{\lambda,T}^{(2)}\big)
\end{align*}
due to the non-negativity of the relative entropy.
\end{proof}

After having considered the partition function, we now prove that essentially $\Gamma_{\lambda,T}^{(1)} \leq C \: \Gamma_{0,T}^{(1)}$ when $\lambda \sim T^{-1}$.

\begin{lemma}[\textbf{Bounds on the one-particle density matrix}]\label{lem:1PDM_bounded}\mbox{}\\
Under the assumptions of Lemma~\ref{lem:unif_bounds}, we have
\begin{equation}
0\leq \Gamma_{\lambda,T}^{(1)}\leq 2T\Big(1+\lambda T\tr_{\gH\otimes_s\gH} (w h^{-1}\otimes h^{-1})\Big)\,h^{-1}.
\label{eq:1PDM_bounded}
\end{equation}
In particular we get
\begin{align} 
\Tr_{\cF(\gH)} \Big[ \dGamma(h^{1-p}) \Gamma \Big] &= \Tr_\gH \Big[ h^{1-p} \Gamma^{(1)} \Big]\nn\\
&\le 2T\Big(1+\lambda T\tr_{\gH\otimes_s\gH} (w h^{-1}\otimes h^{-1})\Big)\tr_\gH(h^{-p}).\label{eq:1PDM-dA}
\end{align}
\end{lemma}

\begin{proof}[Proof of Lemma~\ref{lem:1PDM_bounded}]
The proof is reminiscent of a Feynman-Hellmann argument. Let $x\in D(h)$ be a fixed normalized vector and $A:=h^{1/2}P_xh^{1/2}$ with $P_x=|x\>\< x|$. Note then that $0\leq A\leq h$.
Using Lemma~\ref{lem:unif_bounds}, we write
\begin{align*}
&(\lambda T)\tr \big(w h^{-1}\otimes h^{-1}\big)-\frac{1}{2T}\tr A\Gamma_{\lambda,T}^{(1)}\\
&\qquad\qquad\geq
-\log\frac{Z_\lambda(T)}{Z_0(T)}-\frac{1}{2T}\tr A\Gamma_{\lambda,T}^{(1)}\\
&\qquad\qquad=\cH(\Gamma_{\lambda,T},\Gamma_{0,T})+\frac\lambda{T}\tr \big(w\Gamma_{\lambda,T}^{(2)}\big)-\frac{1}{2T}\tr A\Gamma_{\lambda,T}^{(1)}\\
&\qquad\qquad\geq \inf_{\Gamma'}\left\{\cH(\Gamma',\Gamma_{0,T})-\frac{1}{2T}\tr A(\Gamma')^{(1)}\right\}\\
&\qquad\qquad=\frac1T\tr_\gH \Big(\log(1-e^{-(h-A/2)/T})-\log(1-e^{-h/T})\Big).
\end{align*}
The function $f(x)=-\log(1-e^{-x})$ being convex, we have Klein's inequality 
$$\tr \big(f(A)-f(B)\big)\geq \tr f'(B)(A-B)$$
for any self-adjoint operators $A,B$ (see, e.g.,~\cite[Prop. 3.16]{OhyPet-93} and~\cite[Thm 2.5.2]{Ruelle}). Therefore we obtain
\begin{align*}
(\lambda T)\tr \big(w h^{-1}\otimes h^{-1}\big)-\frac1{2T}\tr A\Gamma_{\lambda,T}^{(1)}&\geq -\frac{1}{2T} \tr \left(A\frac{1}{e^{(h-A/2)/T}-1}\right)\\
&\geq -\frac12\tr \left(A (h-A/2)^{-1}\right)\\
&\geq -\tr \left(A h^{-1}\right)=-1.
\end{align*}

The conclusion is that 
$$\pscal{x,h^{1/2}\frac{\Gamma_{\lambda,T}^{(1)}}{2T}h^{1/2}x}\leq 1+(\lambda T)\tr \big(w h^{-1}\otimes h^{-1}\big)$$
for every normalized $x$, which means that
$$h^{1/2}\frac{\Gamma_{\lambda,T}^{(1)}}{2T}h^{1/2}\leq 1+(\lambda T)\tr \big(w h^{-1}\otimes h^{-1}\big)$$
in the sense of quadratic forms. Multiplying by $h^{-1/2}$ on both sides gives~\eqref{eq:1PDM_bounded}.
\end{proof}

%%%%%%%%%%%%%%%%%%%%%%%%%%%%%%%%%%
\subsubsection*{\bf Step 2: free-energy upper bound.} 

Here we prove the upper bound by means of a trial state argument. We treat the cases $p=1$ and $p>1$ at once.

\begin{lemma}[\textbf{Free-energy upper bound}]\label{lem:upper_bound}\mbox{}\\
Let $h>0$ and $w\geq0$ be two self-adjoint operators on $\gH$ and $\gH\otimes_s\gH$ respectively, such that 
$$\tr_{\gH}(h^{-p})+\tr_{\gH\otimes_s\gH}(w\,h^{-1}\otimes h^{-1})<\ii$$ 
for some $1\leq p<\ii$. Then we have
\begin{equation}
\limsup_{\substack{T\to \infty\\ \lambda T\to1}}-\log\left(\frac{Z_\lambda(T)}{Z_0(T)}\right) \leq -\log(z_r).
\label{eq:upper_bound}
\end{equation}
\end{lemma}

\begin{proof}
Let $J$ be a fixed integer and let $V_J={\rm span}(u_1,...,u_J)$ with associated orthogonal projection $P_J$ in $\gH^{1-p}$, where we recall that the $u_j$'s are the eigenvectors of $h$. We recall the isomorphism of Fock spaces 
$$\cF(\gH)\simeq \cF(V_J)\otimes\cF(V_J^\perp)$$ 
which corresponds to the decomposition 
$$\Gamma_{0,T}\simeq \Gamma_{0,T}^{\leq J}\otimes \Gamma_{0,T}^{> J}$$
with
$$\Gamma_{0,T}^{\leq J}=\prod_{j=1}^J(e^{\lambda_j/T}-1)\, e^{-\bH_0^{\leq J}/T},\qquad \Gamma_{0,T}^{>J}=\prod_{j\geq J+1}(e^{\lambda_j/T}-1)\, e^{-\bH_0^{> J}/T}$$
and an obvious similar notation for the operators $\bH_0^{\leq J}$ and $\bH_0^{>J}$ of $\cF(V_J)$ and $\cF(V_J^\perp)$, respectively. In the same vein, we define
$$\Gamma_{\lambda,T}^{\leq J}=\frac{e^{-\bH_\lambda^{\leq J}/T}}{\tr_{\cF(V_J)}e^{-\bH_\lambda^{\leq J}/T}}$$
with $\bH_\lambda^{\leq J}:=\bH^{\leq J}_0+\lambda \mathbb{W}^{\leq J}$ where $\mathbb{W}^{\leq J}$ is the second quantization of the operator $P_J^{\otimes2}wP_J^{\otimes2}$ in the Fock space $\cF(V_J)$. Now, our trial state is simply
\begin{equation}
\boxed{\Gamma=\Gamma_{\lambda,T}^{\leq J}\otimes \Gamma_{0,T}^{> J},}
\label{eq:trial_state}
\end{equation}
that is, we use the truncated interaction in $V_J$ and the free Gibbs state outside of $V_J$.
Due to the variational principle~\eqref{eq:min_Gibbs}, we have
$$-\log\frac{Z_\lambda(T)}{Z_0(T)}\leq \cH(\Gamma,\Gamma_{0,T})+\frac\lambda T \tr \big[w\Gamma^{(2)}\big].$$
We now estimate the terms on the right side.

\medskip

\noindent\emph{Reduction to a finite dimensional estimate.} For the relative entropy, we use that $\cH(A\otimes B,C\otimes B)=\cH(A,C)$ and deduce 
$$\cH(\Gamma,\Gamma_{0,T})=\cH(\Gamma_{\lambda,T}^{\leq J},\Gamma_{0,T}^{\leq J}).$$
For the interaction term, a calculation shows that 
$$ \Gamma^{(2)} = [\Gamma_{\lambda,T}^{\leq J}]^{(2)} + [\Gamma_{0,T}^{>J}]^{(2)} + \frac{1}{2} \Big( [\Gamma_{\lambda,T}^{\leq J}]^{(1)} \otimes [\Gamma_{0,T}^{>J}]^{(1)} + [\Gamma_{0,T}^{>J}]^{(1)}\otimes [\Gamma_{\lambda,T}^{\leq J}]^{(1)}\Big).$$
Using Lemma~\ref{lem:1PDM_bounded} in the space $\cF(V_J)$ and the facts that 
$$[\Gamma_{0,T}^{>J}]^{(2)}=\left(\frac{1}{e^{h^{>J}/T}-1}\right)^{\otimes 2}\leq T^2h^{-1}(P_J)^\perp\otimes h^{-1}(P_J)^\perp,$$
$$[\Gamma_{0,T}^{>J}]^{(1)}=\frac{1}{e^{h^{>J}/T}-1}\leq Th^{-1}(P_J)^\perp,$$
we conclude that
$$\Gamma^{(2)} \leq [\Gamma_{\lambda,T}^{\leq J}]^{(2)} + T^2C\,\big(h^{-1}\otimes h^{-1}\big)\big(P_J^\perp\otimes P_J^\perp +P_J\otimes P_J^\perp+P_J^\perp\otimes P_J\big)$$
and hence, when $\lambda \sim T ^{-1}$,
$$-\log\frac{Z_\lambda(T)}{Z_0(T)}\leq \cH(\Gamma_{\lambda,T}^{\leq J},\Gamma_{0,T}^{\leq J})+\frac\lambda T \tr w[\Gamma_{\lambda,T}^{\leq J}]^{(2)}+C\tr(wh^{-1}\otimes (P_J^\perp h^{-1})).$$
By definition, $\Gamma_{\lambda,T}^{\leq J}$ minimizes the sum of the first two terms, so we arrive at 
\begin{equation}
-\log\frac{Z_\lambda(T)}{Z_0(T)}\leq -\log\frac{Z_\lambda^{\leq J}(T)}{Z_0^{\leq J}(T)}+C\tr(wh^{-1}\otimes (P_J^\perp h^{-1}))
\label{eq:estim_J}
\end{equation}
with %, of course,
$$Z_\lambda^{\leq J}(T)=\tr_{\cF(V_J)}\left(e^{-\bH_\lambda^{\leq J}/T}\right).$$
The last term on the right side of~\eqref{eq:estim_J} is independent of $T$ and it converges to 0 when $J\to\ii$ since $\tr(wh^{-1}\otimes h^{-1})$ is finite by assumption. But first we are going to take the limit $T\to\ii$.

\medskip

\noindent\emph{Semi-classics in the projected space.}
Finite-dimensional semi-classical analysis predicts that
\begin{align*}
Z^{\leq J}_\lambda(T)&\underset{T\to\ii}\sim \left(\frac{T}\pi\right)^{J}\int_{V_J}e^{-\pscal{u,hu}-\lambda T\pscal{u\otimes u,w u\otimes u}}\,du,\\ \quad Z_0^{\leq J}(T)&\underset{T\to\ii}\sim \left(\frac{T}\pi\right)^{J}\int_{V_J}e^{-\pscal{u,hu}}\,du.
\end{align*}
For $Z_0^{\leq J}(T)$, this can be justified using a direct computation:
\begin{equation} \label{eq:Free-non-int<=}
\frac{Z_0^{\leq J}(T)}{T^J} = \frac{1}{T^J} \prod_{i=1}^J \frac{1}{1-e^{-\lambda_i/T}} \underset{T\to \infty}{\longrightarrow}\prod_{i=1}^J\frac{1}{\lambda_i} = \pi^{-J} \int_{V_J} e^{-\pscal{u,hu}} d u.
\end{equation}
We only need the lower bound on $Z_\lambda^{\leq J}(T)$, i.e. an upper bound on the finite-dimensional free-energy. This is an exact estimate that does not require to take the limit $T\to\ii$. The proof is well-known~\cite{Lieb-73b,Simon-80,LieSeiYng-05} and we quickly discuss it for completeness.

We recall the resolution of the identity~\eqref{eq:resolution_coherent} in terms of coherent states, from which we conclude that
\begin{equation} \label{eq:Tr-Gibbs-coherent}
\tr_{\cF(V_J)}\left[e^{-\bH_\lambda^{\leq J}/T}\right]= \frac{T^J}{\pi^{J}} \int_{V_J} \left\langle 0 \left| W(\sqrt{T}u)^* e^{-\bH_\lambda^{\leq J}/T} W(\sqrt{T}u)\right|0 \right\rangle du. 
\end{equation}
By the Peierls-Bogoliubov inequality $\pscal{x,e^Ax}\geq e^{\pscal{x,Ax}}$, we obtain
\begin{multline} \label{eq:Tr-Gibbs-coherent2}
\tr_{\cF(V_J)}\left[e^{-\bH_\lambda^{\leq J}/T}\right]\geq \\
\left(\frac{T}{\pi}\right)^J \int_{V_J} \exp\left(-\left\langle 0 \left| W(\sqrt{T}u)^* \frac{\bH_\lambda^{\leq J}}{T}W(\sqrt{T}u)\right|0 \right\rangle\right) du. 
\end{multline}
From the property~\eqref{eq:coherent_eigenfn} that coherent states are eigenfunctions of the annihilation operator (or equivalently from the formula~\eqref{eq:DM_coherent} of its density matrices), one can easily show that
$$\forall u\in V_J,\qquad\left\langle 0 \left| W(\sqrt{T}u)^* \bH_\lambda^{\leq J}W(\sqrt{T}u)\right|0 \right\rangle=T \pscal{u,hu}+\frac{\lambda T ^2}2 \pscal{u\otimes u,wu \otimes u}$$
and the lower bound on $Z_\lambda^{\leq J}(T)$ follows immediately.

\medskip 

\noindent\emph{Passing to the limit: $T\to \infty$, then $J\to \infty$.}
In conclusion, we have proved that
\begin{align*}
&-\log\left(\frac{Z_\lambda(T)}{Z_0(T)}\right)\\
%&\quad\leq -\log\left(\frac{Z_\lambda^{\leq J}(T)}{Z_0^{\leq J}(T)}\right)+C\tr(wh^{-1}\otimes (P_J^\perp h^{-1}))\\
&\quad\leq -\log\left(\frac{(T/\pi)^J\int_{V_J}e^{-\cE(u)}\,du}{Z_0^{\leq J}(T)}\right)+C\tr(wh^{-1}\otimes (P_J^\perp h^{-1}))\\
&\quad\underset{T\to\ii}{\longrightarrow} -\log\left(\int_{V_J}e^{-\pscal{u\otimes u,wu \otimes u}/2}d\mu_{0,J}(u)\right)+C\tr(wh^{-1}\otimes (P_J^\perp h^{-1}))
\end{align*}
where
$$\mu_{0,J}:=\prod_{j = 1}^J  \left( \frac{\lambda_j}{\pi}e^{ - {\lambda} _j|\alpha_j|^2}\,d\alpha_j \right)$$
is the cylindrical projection of $\mu_0$  onto $V_J$. Note that
$$\int_{V_J}e^{-\pscal{u\otimes u,wu \otimes u}/2}d\mu_{0,J}(u)=\int_{\gH^{1-p}}e^{-\pscal{u\otimes u,P_J^{\otimes2}wP_J^{\otimes2}u \otimes u}/2}d\mu_{0}(u)$$
converges to $-\log(z_r)$, by the dominated convergence theorem, and~\eqref{eq:upper_bound} follows by taking the limit $J\to\ii$.
\end{proof}

%%%%%%%%%%%%%%%%%%%%%%%%%%%
\subsubsection*{\bf Step 3: lower bound for $p=1$}

Now we explain the proof of the lower bound in the trace-class case $\tr(h^{-1})<\ii$. By Lemma~\ref{lem:1PDM_bounded}, we already know that $\Gamma_{\lambda,T}^{(1)}/T$ is bounded in the trace-class. The following says that the other density matrices are bounded as well.

\begin{lemma}[\textbf{Estimates on higher moments in the trace-class case}]\label{lem:estim_1PDM_high_moments}
Let $h>0$ and $w\geq0$ be two self-adjoint operators on $\gH$ and $\gH\otimes_s\gH$ respectively, such that 
$$\tr_{\gH}(h^{-1})+\tr_{\gH\otimes_s\gH}(w\,h^{-1}\otimes h^{-1})<\ii.$$ 
Then we have
\begin{equation}
\tr_{\cF(\gH)}\big[(\cN/T)^k\Gamma_{\lambda,T}\big]\leq C_ke^{\lambda T\tr(wh^{-1}\otimes h^{-1})}(\tr h^{-1})^k.
\label{eq:bound_PDM}
\end{equation}
In particular, $\Gamma_{\lambda,T}^{(k)}/T^k$ is bounded in the trace-class in the limit where $T\to\ii$ and $\lambda T\to 1$.
\end{lemma}

\begin{proof}
Since $\cN$ commutes with $\bH_\lambda$ and $\mathbb{W}\geq0$, we have
\begin{align*}
\tr_{\cF(\gH)}(\cN/T)^k\Gamma_{\lambda,T}&=Z_\lambda(T) ^{-1} T^{-k}\tr_{\cF(\gH)}\left(\cN^ke^{-\bH_0/T-\lambda\mathbb{W}/T}\right)\\
&\leq Z_\lambda(T) ^{-1} T^{-k}\tr_{\cF(\gH)}\left(\cN^ke^{-\bH_0/T}\right)\\
&= \frac{Z_0(T)}{Z_\lambda(T)}T^{-k}\tr_{\cF(\gH)}\left(\cN^k\Gamma_{0,T}\right)
\end{align*}
and the rest follows from~\eqref{eq:unif_bd2} and the estimates on the density matrices of $\Gamma_{0,T}$ in Section~\ref{sec:non-interacting}.
\end{proof}

We are now able to prove the lower bound on the relative partition function, and the fact that the density matrices of $\Gamma_{\lambda,T}$ all converge to the expected average involving the nonlinear Gibbs measure. Up to extraction of a (not relabeled) subsequence, we may assume from Theorem~\ref{thm:deFinetti} that the sequence $\{\Gamma_{\lambda,T}\}$ admits $\nu$ as de Finetti measure, and that 
$$k! T^{-k}\Gamma_{\lambda,T}^{(k)}\wto_* \int_{\gH}|u^{\otimes k}\>\< u^{\otimes k}|\,d\nu(u)$$
weakly-$\ast$ in the trace class, for all $k\geq1$. Note that since we have the uniform upper bound $\Gamma_{\lambda,T}^{(1)}/T\leq Ch^{-1}$ by Lemma~\ref{lem:1PDM_bounded}, the convergence must indeed be strong in $\gS^1$ for the first density matrix, by the dominated convergence theorem in the trace-class~\cite[Thm 2.16]{Simon-79}.
We recall that
$$\liminf_{\substack{T\to\ii\\ \lambda T\to1}}\left(-\log\frac{Z_\lambda(T)}{Z_0(T)}\right)=\liminf_{T\to\ii}\left(\cH(\Gamma_{\lambda,T},\Gamma_{0,T})+\frac{\lambda T}{T^2}\tr(w\Gamma_{\lambda,T}^{(2)})\right).$$
Since $w\geq 0$ we have, by Fatou's lemma for operators and the de Finetti representation formula for the weak limit of $\Gamma_{\lambda,T}^{(2)}$, 
\begin{align*}
\liminf_{T\to\ii}\frac{1}{T ^2}\tr(w\Gamma_{\lambda,T}^{(2)})&\geq \frac12\tr\left(w\int_{\gH}|u^{\otimes 2}\>\<u^{\otimes 2}|d\nu(u)\right)\\
&=\frac12\int_{\gH}\pscal{u\otimes u,wu\otimes u}d\nu(u).
\end{align*}
Using then Theorem~\ref{thm:rel-entropy} for the relative entropy, we get
$$\liminf_{\substack{T\to\ii\\ \lambda T\to1}}\left(-\log\frac{Z_\lambda(T)}{Z_0(T)}\right)\geq\cHcl(\nu,\mu_0)+\frac12\int_{\gH}\pscal{u\otimes u,wu\otimes u}d\nu(u).$$
From the variational principle~\eqref{eq:min_Gibbs_nonlinear} the right side is bounded from below by $-\log(z_r)$. Since we have already proved the upper bound in the previous section and since $\mu$ is the only minimizer, we conclude that $\nu=\mu$ and that
$$\lim_{\substack{T\to\ii\\ \lambda T\to1}}\left(-\log\frac{Z_\lambda(T)}{Z_0(T)}\right)=-\log(z_r).$$
In particular, we have the weak convergence of all the density matrices to the desired integral representation, and the strong convergence (in the trace-class) for $k=1$. 

Except for the proof of the strong convergence for $k\geq2$ (which we postpone for the moment), this ends the proof of Theorem~\ref{thm:main}.

%%%%%%%%%%%%%%%%%%%%%%%%%%%
\subsubsection*{\bf Step 4: lower bound for $p>1$}

In the case $p>1$, the main difficulty is the lack of control on the higher density matrices. We know from~\eqref{eq:estim_interaction} that $\tr(w\Gamma_{\lambda,T}^{(2)})$ is bounded by $CT^2$, but we do not have a bound on $\Gamma_{\lambda,T}^{(2)}$ itself. 

The beginning of the argument is the same as for $p=1$. By Lemma~\ref{lem:1PDM_bounded}, we already know that $\Gamma_{\lambda,T}^{(1)}/T$ is bounded in $\gS^p$ and that $\Gamma_{\lambda,T}^{(1)}\leq Th^{-1}$. After extraction of a subsequence and by the dominated convergence theorem, we may assume 
$$T ^{-1}\Gamma_{\lambda,T}^{(1)}\underset{T\to\ii}{\longrightarrow}\gamma^{(1)}$$
strongly in $\gS^p$, for some limiting operator $\gamma^{(1)}\in\gS^p$. By Theorem~\ref{thm:deFinetti}, we may also assume (after extraction of a subsequence), that $\{\Gamma_{\lambda,T}\}$ admits $\nu$ as de Finetti measure on $\gH^{1-p}$. Hence the lower symbols 
$$\mu_{V} ^T:= \mu_{V,\Gamma_{\lambda,T}}^{1/T}$$
defined as in Section~\ref{sec:deFinetti_coherent} converge to $\nu_V$, for every fixed finite-dimensional subspace $V$.
Now, in order to mimic the argument used for $p=1$, we need the following 

\begin{proposition}[\textbf{The interaction is weakly lower semi-continuous}]\label{prop:w_wlsc}
We have
\begin{equation}
\liminf_{T\to\ii}\frac{\tr[w\Gamma_{\lambda,T}^{(2)}]}{T^2}\geq \frac12\int_{\gH^{1-p}} \pscal{u^{\otimes 2},wu^{\otimes 2}}\,d\nu(u).
\label{eq:w_wlsc}
\end{equation}
\end{proposition}

The purpose of the rest of this step is to prove Proposition~\ref{prop:w_wlsc}. The following summarizes some estimates that we have on the higher density matrices.

\begin{lemma}[\textbf{Moment estimates in the general case}]\label{lem:estim_1PDM_high_moments_p}\mbox{}\\
Let $h>0$ and $w\geq0$ be two self-adjoint operators on $\gH$ and $\gH\otimes_s\gH$ respectively, such that 
$$\tr_{\gH}(h^{-p})+\tr_{\gH\otimes_s\gH}(w\,h^{-1}\otimes h^{-1})<\ii$$ 
for some $1\leq p<\ii$. Then we have, for all $s\geq1$ and all $\lambda T\leq C$,
\begin{equation}
\tr_{\cF(\gH)}\big[\cN^s\Gamma_{\lambda,T}\big]\leq C_s\, T^{ps}
\label{eq:estim_N_s_non-trace-class}
\end{equation}
and for all $\epsilon>0$
\begin{equation}
\tr_{\cF(\gH)}\big[\cN\dGamma(h^{1-p})\Gamma_{\lambda,T}\big]\leq C_\epsilon\, T^{p+1+\epsilon}.
\label{eq:estim_mixed}
\end{equation}
\end{lemma}

\begin{proof}
The same argument as in Lemma~\ref{lem:estim_1PDM_high_moments} gives for $s$ an integer
\begin{align*}
\tr_{\cF(\gH)}\big[\cN^s\Gamma_{\lambda,T}\big]&\leq C\tr_{\cF(\gH)}\big[\cN^s\Gamma_{0,T}\big]\\
&\leq C\big[\tr_\gH(e^{h/T}-1)^{-1}\big]^s\\
&\leq C\big[\tr_\gH(h/T)^{-p}\big]^s=CT^{ps}\big[\tr_\gH h^{-p}]^{s}.
\end{align*}
The argument follows for all $s\geq1$ by interpolation. For the other estimate, we use Hölder's inequality (and the fact that $\cN$ commutes with $d\Gamma(h^{1-p})$) to obtain
\begin{align*}
\Tr \Big[ \dGamma (h^{1-p}) \cN \Gamma _{\lambda,T}\Big] & \le \Big(  \Tr \big[ \dGamma (h^{1-p}) \Gamma_{\lambda,T} \big] \Big)^{1-\theta} \Big(  \Tr \big[ \dGamma (h^{1-p}) \cN^{1/\theta} \Gamma_{\lambda,T} \big] \Big)^{\theta} \\
& \le C_\theta\Big(  \Tr \big[ \dGamma (h^{1-p}) \Gamma_{\lambda,T} \big] \Big)^{1-\theta} \Big(  \Tr \big[ \cN^{(1+\theta)/\theta} \Gamma_{\lambda,T} \big] \Big)^{\theta} \\
& \le C_\theta T^{1-\theta} T^{p(1+\theta)} = C_\eps T^{p+1+\eps} 
\end{align*}
for all $\theta\in (0,1)$, where $\eps=\theta(p-1)>0$. In the last line we have used~\eqref{eq:1PDM_bounded}.
\end{proof}

We conjecture that the bound $\tr[(\dGamma(h^{1-p}))^2\Gamma_{\lambda,T}]\leq CT^2$ holds true. 

Now, the idea of the proof of Proposition~\ref{prop:w_wlsc} is to localize the problem to the finite dimensional space $V_J={\rm span}(u_1,...,u_J)$ with a $J=J(T)$ that grows at a convenient speed in order to control the errors in the localization procedure, which can be estimated using the bounds of Lemma~\ref{lem:estim_1PDM_high_moments_p}. 

\begin{lemma}[\textbf{Localization of the interaction}]\mbox{}\\
Let $h>0$ and $w\geq0$ be two self-adjoint operators on $\gH$ and $\gH\otimes_s\gH$ respectively, such that $\tr_{\gH}(h^{-p})<\ii$ and $w\leq Ch^{1-p'}\otimes h^{1-p'}$ for some $p'>p\geq 1$. Then we have 
\begin{equation}
\tr\big[w\Gamma_{\lambda,T}^{(2)}\big]\geq \tr\big[(P_J)^{\otimes 2}w(P_J)^{\otimes 2}\Gamma_{\lambda,T}^{(2)}\big]- C_\eps \frac{  T^{p+1+\epsilon}}{(\lambda_{J+1})^{p'-1}}-C_\eps\frac{T^{\frac{p+3+\eps}{2}}}{(\lambda_{J+1})^{\frac{p'-1}{2}}}
%-C_\eps\frac{T^{p+1+\eps}}{(\lambda_{J+1})^{p'-1}}
\label{eq:estim_error_localization}
\end{equation}
for $\lambda T\leq C$, $T\geq 1/C$ and any $\eps >0$.
\end{lemma}

\begin{proof}
For simplicity we denote $P=P_J$, $Q=1-P_J$ and 
$$\Pi=1-(P_J)^{\otimes 2}=Q\otimes P+P\otimes Q+Q\otimes Q.$$
As in the proof of~\cite[Lemma 3.6]{LewNamRou-14c} we write
$$w=P^{\otimes 2}wP^{\otimes 2}+\Pi w P^{\otimes 2}+P^{\otimes 2} w \Pi +\Pi w\Pi$$
and note that, since $w\geq 0$, 
$$ \eta P^{\otimes 2}wP^{\otimes 2} +  \Pi w P^{\otimes 2}+P^{\otimes 2} w \Pi + \eta ^{-1} \Pi w \Pi \ge 0 $$
for any $\eta >0$. Applying the latter estimate with $\eta$ replaced by $\eta/(1+\eta)$ we find that 
\begin{align*}
&(1+\eta)w-P^{\otimes 2} w P^{\otimes 2} + (1+\eta^{-1}) \Pi w \Pi  \\
=&\eta P^{\otimes 2} w P^{\otimes 2} + (1+\eta) ( \Pi w P^{\otimes 2}+P^{\otimes 2} w \Pi ) +  (1+\eta)^2 \eta^{-1} \Pi w \Pi \ge 0
\end{align*}
for any $\eta >0$. By assumption $0\leq w\leq C h^{1-p'}\otimes h^{1-p'}$, we deduce
\begin{align*}
w-P^{\otimes 2}w P^{\otimes 2}&\geq -  \eta w- (1+\eta^{-1}) \Pi w \Pi\\
&\geq -\eta w- C(1+\eta^{-1}) (\lambda_{J+1})^{1-p'}\big(h^{1-p}\otimes \1+\1\otimes h^{1-p}\big),
\end{align*}
where we have used $Q h^{1-p'}\le (\lambda_{J+1})^{1-p'}$ and $h^{1-p'}\le C h^{1-p}$ in the last estimate. Consequently, for any state $\Gamma$ and for any $\eta>0$ 
\begin{multline*}
\Tr \Big[ (w-P^{\otimes 2}w P^{\otimes 2}) \Gamma^{(2)} \Big] \ge -\eta \Tr \Big[ w \Gamma^{(2)} \Big]\\  - C (1+\eta^{-1}) (\lambda_{J+1})^{1-p'}  \Tr \Big[ \dGamma (h^{1-p}) \cN \Gamma \Big].
\end{multline*}
Applying this inequality to $\Gamma_{\lambda,T}$ and inserting the bounds~\eqref{eq:estim_interaction} and~\eqref{eq:estim_mixed} on its density matrices, we get
\begin{multline*}
\Tr \Big[ (w-P^{\otimes 2}w P^{\otimes 2}) \Gamma^{(2)} \Big] \ge -C\eta T^2 - C_\eps (1+\eta^{-1}) (\lambda_{J+1})^{1-p'}  T^{p+1+\epsilon}
\end{multline*}
and the result follows after optimizing with respect to~$\eta$.
\end{proof}

We are finally able to provide the

\begin{proof}[Proof of Proposition~\ref{prop:w_wlsc}]
Let $\Gamma^{\leq J}_{\lambda,T}$ be the localized state of $\Gamma_{\lambda,T}$ with respect to the projection $P_J$ as recalled in Section~\ref{sec:Fock2} and used several times before. By definition, we have
$
[\Gamma^{\leq J}_{\lambda,T}]^{(2)}=  (P_J)^{\otimes 2} \Gamma_{\lambda,T}^{(2)} (P_J)^{\otimes 2}.
$
Let $\mu_{V_J}^T$ be the corresponding Husimi function defined as in \eqref{eq:Husimi}. By Lemma~\ref{lem:link_PDM_Husimi}, we have 
\begin{multline*}
\frac{T^2}2\int_{V_J}|u^{\otimes 2}\>\<u^{\otimes 2}|\;d\mu_{V_J}^T(u)=(P_J)^{\otimes 2}\Gamma_{\lambda,T}^{(2)}(P_J)^{\otimes 2}+(P_J)^{\otimes 2}\\
+P_J\otimes(P_J\Gamma_{\lambda,T}^{(1)}P_J)+(P_J\Gamma_{\lambda,T}^{(1)}P_J)\otimes P_J
\end{multline*}
and hence
\begin{multline*}
\tr\big[(P_J)^{\otimes 2}w(P_J)^{\otimes 2}\Gamma_{\lambda,T}^{(2)}\big]\geq \frac{T^2}2\int_{V_J} \pscal{u^{\otimes 2},wu^{\otimes 2}}\,d\mu_{V_J}^T(u)\\-\tr\big[w(P_J)^{\otimes 2}\big]-2\tr\big[(P_J)^{\otimes 2}w(P_J)^{\otimes 2}\1\otimes \Gamma_{\lambda,T}^{(1)}\big].
\end{multline*}
Now we use the assumption that 
\begin{equation}\label{eq:recall bound w}
w\leq C h^{1-p'}\otimes h^{1-p'} 
\end{equation}
and the bound $\Gamma_{\lambda,T}^{(1)}\leq C T h^{-1}$ from Lemma~\ref{lem:1PDM_bounded} and we get 
\begin{align*}
\frac1{T ^2}\tr\big[(P_J)^{\otimes 2}w(P_J)^{\otimes 2}\Gamma_{\lambda,T}^{(2)}\big]&\geq \frac{1}2\int_{V_J} \pscal{u^{\otimes 2},wu^{\otimes 2}}\,d\mu_{V_J}^T(u)\\
&\qquad-\frac{\Big(\tr\big[h^{1-p'}P_J\big] \Big)^2}{T^2}-C\frac{\tr\big[h^{1-p'}P_J\big]}{T}\\
&\geq \frac{1}2\int_{V_J} \pscal{u^{\otimes 2},wu^{\otimes 2}}\,d\mu_{V_J}^T(u)\\
&\qquad-C\frac{(\lambda_{J+1})^{2(1+p-p')}}{T^2}-C\frac{(\lambda_{J+1})^{1+p-p'}}{T}.
\end{align*}
If we combine this with~\eqref{eq:estim_error_localization}, we find 
\begin{align*}
T^{-2}\tr\big[w\Gamma_{\lambda,T}^{(2)}\big] \geq& \frac{1}2\int_{V_J} \pscal{u^{\otimes 2},wu^{\otimes 2}}\,d\mu_{V_J}^T(u)\\
&-C\frac{(\lambda_{J+1})^{2(1+p-p')}}{T^2}-C\frac{(\lambda_{J+1})^{1+p-p'}}{T}\\
&-C_\eps \frac{  T^{p-1+\epsilon}}{(\lambda_{J+1})^{p'-1}}-C_\eps\frac{T^{\frac{p-1+\eps}{2}}}{(\lambda_{J+1})^{\frac{p'-1}{2}}}.
\end{align*}
To make all four error terms small, we choose $\lambda_{J+1}\sim T^\alpha$, which amounts to picking a suitably large $J \to \infty$ when $T\to \infty$, with $\alpha>0$ satisfying
$$
\alpha(1+p-p')<1 \quad \text{and} \quad p-1+\eps<\alpha(p'-1).
$$
This can be done when $\eps>0$ is small enough since
$$
\frac{p'-1} {p-1} > 1+p-p' 
$$
where we have used $p'>p \ge 1$.  

Now it only remains to prove that
$$\liminf_{n\to\ii}\int_{V_{J_n}} \pscal{u^{\otimes 2},wu^{\otimes 2}}\,d\mu_{V_J}^T(u)\geq \int_{\gH^{1-p}} \pscal{u^{\otimes 2},wu^{\otimes 2}}\,d\nu(u).$$
The main observation is that 
$$\mu_{V_J}^T\wto \nu$$
in the sense that each cylindrical projection to a \emph{fixed} finite dimensional subspace $V'$ converges in the sense of measures on $V'$. The reason is that, due to Lemma~\ref{lem:cylindrical}, the cylindrical projection of $\mu_{V_J}^T$ onto $V'$ is $\mu^{T}_{V'}$, which converges to $\nu_{V'}$ by definition of $\nu$ in the proof of Theorem~\ref{thm:deFinetti}. Therefore, by Fatou's lemma we obtain
$$\liminf_{n\to\ii} \int_{V_{J_n}} \pscal{u^{\otimes 2},wu^{\otimes 2}}\,d\mu_{V_J}^T(u)\geq \int_{\gH^{1-p}} \pscal{u^{\otimes 2},wu^{\otimes 2}}\,d\nu(u)$$
and this concludes the proof of Proposition~\ref{prop:w_wlsc}.
\end{proof}

To summarize, in the case $p>1$ we have proved that $\mu$ is the de Finetti measure of $\{\Gamma_{\lambda,T}\}$ as well as the convergence of the relative partition function. We also know that $\Gamma_{\lambda,T}^{(1)}/T$ converges to a limit $\gamma^{(1)}$ strongly in $\gS^p(\gH)$, but we still have to prove that
\begin{equation}
\gamma^{(1)}=\int_{\gH^{1-p}}|u\>\<u|\,d\mu(u),
\label{eq:identify_limit_1PDM}
\end{equation}
which will conclude the proof of Theorem~\ref{thm:main2}. Unfortunately, we do not know that $\|u\|_{\gH^{1-p}}^2d\mu^{1/T}_{V_{J_n}}$ is tight since we do not have any control on a higher moment. Therefore, we cannot pass to the weak limit in the one-particle density matrix so easily. We need another argument which will be provided in the next step.

%%%%%%%%%%%%%%%%%%%%%%%%
\subsubsection*{\bf Step 5: strong convergence for $p>1$ and $k=1$}

Here we prove~\eqref{eq:identify_limit_1PDM} by using the Feynman-Hellmann principle. The rationale is that we may perturb $h$ by a large class of one-body operators and estimate the free energy in essentially the same way. Differentiating with respect to the perturbation gives access to the one-body density matrix. We cannot apply the same idea to higher density matrices since we require strong assumptions, in particular positivity, of the two-body interaction.

Let $A=|x\>\<x|$ with $x\in V_J$ for some fixed $J$ and $\norm{x}_{\gH}=1$. For 
$$|\eta|<1/\|h^{-1}\|$$
we consider the Gibbs state
$$\Gamma_{\lambda,T,\eta}:=Z_{\lambda,\eta}(T)^{-1}e^{-\tfrac{\dGamma(h+\eta A)+\lambda\mathbb{W}}{T}},\quad Z_{\lambda,\eta}(T)=\tr_{\cF\gH)}\left[e^{-\tfrac{\dGamma(h+\eta A)+\lambda\mathbb{W}}{T}}\right].$$
The condition $|\eta|<1/\|h^{-1}\|$ ensures that $h+\eta A$ is invertible and, since
$(h+\eta A)^{-1}=h^{-1}\big(1+ \eta Ah^{-1}\big)^{-1},$
we conclude that $(h+\eta A)^{-1}\in\gS^p(\gH)$, that is, $\tr(h+\eta A)^{-p}<\ii$. We can therefore apply our results with $h$ replaced by $h+\eta A$ throughout.
We start by writing 
\begin{align}
\frac\eta{T}\tr_\gH\big[A\Gamma_{\lambda,T}^{(1)}\big]&=\frac1T\tr\big[\dGamma(h+\eta A)\Gamma_{\lambda,T}\big]+\frac\lambda{T}\tr\big[\mathbb{W}\Gamma_{\lambda,T}\big]\nn\\
&\qquad +\tr\big[\Gamma_{\lambda,T}\log\Gamma_{\lambda,T}\big]+\log\tr Z_\lambda(T)\nn\\
&\geq -\log\frac{Z_{\lambda,\eta}(T)}{Z_\lambda(T)}\nn\\
&=-\log\frac{Z_{\lambda,\eta}(T)}{Z_{0,\eta}(T)}-\log \frac{Z_{0}(T)}{Z_{\lambda}(T)}-\log\frac{Z_{0,\eta}(T)}{Z_{0}(T)}.\label{eq:lower_bd_FH}
\end{align}
We remark that $h+\eta A=h$ on $V_J^\perp$ and, therefore, by~\eqref{eq:partition_quasi_free} we have
\begin{multline*}
-\log\frac{Z_{0,\eta}(T)}{Z_{0}(T)}=\log\frac{\prod_{j=1}^J(1-e^{-\lambda_j(\eta)/T})}{\prod_{j=1}^J(1-e^{-\lambda_j/T})}\\
\underset{T\to\ii}{\longrightarrow} -\log\left(\int_{\gH^{1-p}}e^{-\eta|\<x,u\>|^2}\,d\mu_0(u)\right).
\end{multline*}
Passing to the limit $T\to\ii$ in~\eqref{eq:lower_bd_FH} and using the convergence of the relative partition function as well as $\Gamma_{\lambda,T}^{(1)}/T\to\gamma^{(1)}$, we deduce that
\begin{equation*}
\eta\tr_\gH\big[A\gamma^{(1)}\big]\geq -\log\frac{\displaystyle\int_{\gH^{1-p}}e^{-\eta|\<x,u\>|^2-\pscal{u\otimes u,w u\otimes u}/2}\,d\mu_0(u)}{\displaystyle\int_{\gH^{1-p}}e^{-\pscal{u\otimes u,w u\otimes u}/2}\,d\mu_0(u)}.
\end{equation*}
Now, taking $\eta\to0^+$ and then $\eta\to0^-$, we conclude that 
\begin{align*}
\tr_\gH\big[A\gamma^{(1)}\big]&=-\lim_{\eta\to0} \eta^{-1}\log\frac{\displaystyle\int_{\gH^{1-p}}e^{-\eta|\<x,u\>|^2-\pscal{u\otimes u,w u\otimes u}/2}\,d\mu_0(u)}{\displaystyle\int_{\gH^{1-p}}e^{-\pscal{u\otimes u,w u\otimes u}/2}\,d\mu_0(u)}\\
&=\frac{\displaystyle\int_{\gH^{1-p}}|\<x,u\>|^2e^{-\pscal{u\otimes u,w u\otimes u}/2}\,d\mu_0(u)}{\displaystyle\int_{\gH^{1-p}}e^{-\pscal{u\otimes u,w u\otimes u}/2}\,d\mu_0(u)}=\int_{\gH^{1-p}}\pscal{u,Au}\,d\mu(u).
\end{align*}
Since $x\in V_J$ is arbitrary, this proves that
$$P_J\gamma^{(1)}P_J=P_J\left(\int_{\gH^{1-p}}|u\>\<u|\,d\mu(u)\right)P_J.$$
Passing to the limit $J\to\ii$, we obtain~\eqref{eq:identify_limit_1PDM} and this concludes the proof of Theorem~\ref{thm:main2}.

%%%%%%%%%%%%%%%%%%%%%%%%
\subsubsection*{\bf Step 6: strong convergence for $p=1$ and $k\geq2$}
Finally, we prove the strong convergence
\begin{align*} 
 \frac{k!}{T^k} \Gamma_{{\lambda},T}^{(k)} \underset{\substack{T\to\ii\\ T\lambda\to1}}{\longrightarrow} \int_{\gH} |u^{\otimes k}\rangle \langle u^{\otimes k}| d\mu(u)
\end{align*}
in the trace class for every $k\geq2$, when $p=1$. We recall that the uniform bound in~\eqref{eq:1PDM_bounded} already gave us the result for $k=1$, and that we know the \emph{weak-$\ast$} convergence for all $k\geq1$.We therefore need to show that
\begin{equation*} %\label{eq:cv-Tr-Gammak}
\limsup_{\substack{T\to\ii\\ T\lambda\to1}} \frac{k!}{T^k} \Tr_{\gH^k} \Big( \Gamma_{{\lambda},T}^{(k)} \Big) \le \int_{\gH} \|u\|_\gH^{2k} \,d\mu(u)
\end{equation*} 
and, since
\begin{equation*} 
\frac{k!}{T^k} \Tr_{\gH^k} \Big(\Gamma_{\lambda,T}^{(k)} \Big) = \frac{1}{T^k}\tr_{\F(\gH)} \Big( \cN(\cN-1)...(\cN-k+1) \Gamma_{\lambda,T}\Big),
\end{equation*}
it suffices to prove that
\begin{align} \label{eq:upper-bound-Tr-Gammak}
\limsup_{\substack{T\to\ii\\ T\lambda\to1}} \Tr_{\F(\gH)} \left( \frac{\cN^k}{T^k} \Gamma_{\lambda,T} \right) \le \int_{\gH} \|u\|_\gH^{2k} \,d\mu(u)
\end{align} 
for every $k\in \mathbb{N}$. We write
\begin{equation*}
\Tr \left( \frac{\cN^k}{T^k} \Gamma_{\lambda,T} \right)=\frac{\Tr \left( \frac{\cN^k}{T^k} e^{-\bH_\lambda/T} \right)}{\tr\left( \frac{\cN^k}{T^k} e^{-\bH_0/T} \right)}\cdot \frac{\Tr\left( \frac{\cN^k}{T^k} e^{-\bH_0/T} \right)}{\Tr\left(e^{-\bH_0/T} \right)}\cdot \frac{\Tr\left( e^{-\bH_0/T} \right)}{\Tr\left(e^{-\bH_\lambda/T} \right)}
\end{equation*}
and we recall that
$$\lim_{\substack{T\to\ii\\ T\lambda\to1}}\frac{\Tr\left( e^{-\bH_0/T} \right)}{\Tr\left(e^{-\bH_\lambda/T} \right)}=\frac{1}{z_r}$$
by Theorem~\ref{thm:main} and that
$$\lim_{T\to\ii}\frac{\Tr\left( \frac{\cN^k}{T^k} e^{-\bH_0/T} \right)}{\Tr\left(e^{-\bH_0/T} \right)}=\int_{\gH}\|u\|^{2k}_\gH\,d\mu_0(u)$$
by the strong convergence in the trace-class of $\Gamma^{(k)}_{0,T}$ in Lemma~\ref{lem:CV_free}.
Therefore, it remains to prove that
\begin{equation}
\limsup_{\substack{T\to\ii\\ T\lambda\to1}}\frac{\Tr \left( \frac{\cN^k}{T^k} e^{-\bH_\lambda/T} \right)}{\tr\left( \frac{\cN^k}{T^k} e^{-\bH_0/T} \right)}\leq \int_{\gH}e^{-\pscal{u\otimes u,w u\otimes u}}\,d\tilde\mu_0(u)
\label{eq:to_be_shown_PDM_strong}
\end{equation}
where
$$d\tilde\mu_0(u):=\frac{\|u\|^{2k}_{\gH}d\mu_0(u)}{\displaystyle\int_{\gH}\|u\|^{2k}_{\gH}d\mu_0(u)}.$$

The rest of the argument is now exactly the same as in Step 3 and we only sketch the proof. We have
\begin{equation*}
-\log\frac{\Tr \left( \frac{\cN^k}{T^k} e^{-\bH_\lambda/T} \right)}{\tr\left( \frac{\cN^k}{T^k} e^{-\bH_0/T} \right)}=\cH(\tilde\Gamma_{\lambda,T},\tilde\Gamma_{0,T})+\frac{1}{T^2}\tr[w\tilde\Gamma_{\lambda,T}^{(2)}]
\end{equation*}
with
$$
\widetilde{\Gamma}_{\lambda,T}:= \frac{(\cN/T)^k e^{-\bH_\lambda/T}}{ \tr_{\F(\gH)} \big[ (\cN/T)^k e^{-\bH_\lambda/T}\big]}
$$
(recall that $\cN$ commutes with $\bH_\lambda$).
We remark that by the weak convergence of $\Gamma_{\lambda,T}^{(k)}$ we have
$$\liminf_{\substack{T\to\ii\\ T\lambda\to1}}\frac{\tr_{\F(\gH)} \big[ (\cN/T)^k e^{-\bH_\lambda/T}\big]}{\tr_{\F(\gH)} \big[e^{-\bH_\lambda/T}\big]}\geq \int_{\gH}\norm{u}_\gH^{2k}d\mu(u)>0.$$
Then, the bounds in Lemma~\ref{lem:estim_1PDM_high_moments} tell us that 
$$\tr\big[(\cN/T)^{\ell}\tilde\Gamma_{\lambda,T}\big]\leq C_\ell$$
for all $\ell\geq1$. By Theorem~\ref{thm:deFinetti}, we can therefore consider a sequence $T\to\ii$ and $\lambda\sim1/T$ for which $\tilde\Gamma_{T,\lambda}$ admits $\tilde\nu$ as de Finetti measure. 

\begin{lemma}[\textbf{Higher moments and de Finetti measure}]\label{lem:deFinetti_higher_moments}\mbox{}\\
The sequence of states
$$
\widetilde{\Gamma}_{0,T}:= \frac{(\cN/T)^k e^{-\bH_0/T}}{ \tr_{\F(\gH)} \big[ (\cN/T)^k e^{-\bH_0/T}\big]}
$$
admits $\tilde\mu_0$ as de Finetti measure at scale $1/T$.
\end{lemma}

The proof of the lemma is provided in Appendix~\ref{sec:proof_lemma_N_k}.
Using Lemma~\ref{lem:deFinetti_higher_moments}, we get as before
\begin{align*}
\liminf_{\substack{T\to\ii\\ T\lambda\to1}}-\log\frac{\Tr \left( \frac{\cN^k}{T^k} e^{-\bH_\lambda/T} \right)}{\tr\left( \frac{\cN^k}{T^k} e^{-\bH_0/T} \right)}&\geq \cHcl(\tilde\nu,\tilde\mu_0)+\frac12\int_{\gH}\pscal{u\otimes u ,w u\otimes u}d\tilde\nu(u)\\
&\geq -\log\left(\int_{\gH}e^{-\pscal{u\otimes u,w u\otimes u}}\,d\tilde\mu_0(u)\right)
\end{align*}
which concludes the proof of~\eqref{eq:to_be_shown_PDM_strong}, hence of the strong convergence of the $k$-particle density matrices and that of Theorem~\ref{thm:main}.\qed

%%%%%%%%%%%%%%%%%%%%%%%%%%%%%%%%%%%%%%%%%%
%%%%%%%%%%%%%%%%%%%%%%%%%%%%%%%%%%%%%%%%%%
\appendix

\section{Free Gibbs states}\label{sec:proof_quasi_free}

Lemma~\ref{lem:quasi-free} follows from a well-known computation (Wick's theorem). It is convenient to write $a_i:=a(u_i)$ and to use some algebraic properties of quasi free states. First we write $\gH=\bigoplus_{i\geq1}(\C u_i)$ and, as we have recalled above, this implies $\cF(\gH)\simeq \bigotimes_{i\geq1}\cF(\C u_i)$. Next we use that $\dGamma(h)$ can be expressed using creation and annihilation operators as
$$\dGamma(h)=0 \oplus \bigoplus_{m=1}^\infty \left(\sum_{i=1}^m h_i\right)=\sum_{i\geq1}\lambda_i a_i^\dagger a_i.$$
We deduce that $\exp(-\dGamma(h)/T)\simeq\bigotimes_{i\geq1}\exp(-\lambda_i a_i^\dagger a_i/T)$. In the Fock space $\cF(\C u_i)$, we have simply $a_i^\dagger a_i=0\oplus 1\oplus2\oplus\cdots$, the number operator. In particular, we find that
$$\tr_{\cF(\C u_i)}\left[\exp(-\lambda_i a_i^\dagger a_i/T)\right]=\sum_{n\geq0}e^{-\lambda_in /T}=\frac{1}{1-e^{-\lambda_i /T}},$$
and this gives 
$$
\tr_{\cF(\gH)}\left[\exp\left(-\frac{\dGamma(h)}{T}\right)\right]=\prod_{i\geq1} \tr_{\cF(\C u_i)}\left[\exp\left(-\frac{\lambda_i a_i^\dagger a_i}{T}\right)\right]=\prod_{i\geq1}\frac{1}{1-e^{-\lambda_i /T}}.
$$
We also obtain
$$\Gamma_{0,T}:=\frac{e^{-\dGamma(h)/T}}{\tr_{\cF(\gH)}\big(e^{-\dGamma(h)/T}\big)}=\bigotimes_{i\geq1}(1-e^{-\lambda_i /T})\exp\left(-\frac{\lambda_i a_i^\dagger a_i}{T}\right).$$
Now we compute the density matrices of $\Gamma_{0,T}$, using that
\begin{align*}
\tr_{\cF(\C u_i)}\left[(a_i^\dagger)^m(a_i)^{m'}\exp(-\lambda_i a_i^\dagger a_i/T)\right]&=\delta_{mm'}\sum_{n\geq m}e^{-\lambda_in/T}\norm{(a_i)^m(u_i)^{\otimes n}}^2\\
&=\delta_{mm'}\sum_{n\geq m}e^{-\lambda_in/T}\frac{n!}{(n-m)!}\\
&=\frac{\delta_{mm'}\, m!}{\big(e^{\lambda_i/T}-1\big)^{m}\big(1-e^{-\lambda_i/T}\big)}.
\end{align*}
By definition of $\Gamma^{(k)}_{0,T}$, we deduce that 
\begin{equation*}
\frac{\pscal{u_{j_1}\otimes_s \cdots \otimes_s u_{j_k},\Gamma_{0,T}^{(k)}u_{i_1}\otimes_s \cdots \otimes_s u_{i_k}}}{ \norm{u_{i_1}\otimes_s \cdots \otimes_s u_{i_k}}^2}= \prod_{\ell=1}^k\frac{\delta_{i_\ell j_\ell}}{e^{\lambda_{i_\ell}/T}-1}.
\end{equation*}
Therefore we have shown that
\begin{align*}
\Gamma_{0,T}^{(k)}&=\sum_{i_1\leq i_2\leq\cdots\leq i_k}\left(\prod_{\ell=1}^k\frac{1}{e^{\lambda_{i_\ell}/T}-1}\right)\frac{|u_{i_1}\otimes_s \cdots \otimes_s u_{i_k}\rangle\langle u_{i_1}\otimes_s \cdots \otimes_s u_{i_k}|}{\norm{u_{i_1}\otimes_s \cdots \otimes_s u_{i_k}}^2}\nn\\
&=\left(\frac{1}{e^{h/T}-1}\right)^{\otimes k}
\end{align*}
which concludes the proof of Lemma~\ref{lem:quasi-free}.
\qed

\section{de Finetti measure for higher moments}\label{sec:proof_lemma_N_k}

Here we prove Lemma~\ref{lem:deFinetti_higher_moments}. Since we already know from Lemma~\ref{lem:CV_free} that 
\begin{align} \label{eq:strong-cv-Gamma0T}
\lim_{T\to \infty}  \frac{\Tr \left( (\cN/T)^k e^{-\bH_0/T} \right)}{\Tr e^{-\bH_0/T}} = \lim_{T\to \infty} \frac{k!}{T^k} \Tr \Big( \Gamma_{0,T}^{(k)}\Big)  =\int_{\gH} \|u\|^{2k} d\mu_0(u)
\end{align}
for every $k,\ell\in \mathbb{N}$ and since $\tilde\Gamma_{0,T}$ commutes with $\cN$, it suffices to prove that 
\begin{equation}
\lim_{T\to\ii} \Tr \left[ \prod_{i=1}^\ii \left(\frac{a^\dagger_{i}a_{i}}{T}\right)^{n_i}\Big( \frac{\cN}{T} \Big)^k \Gamma_{0,T} \right]\\
=\int_\gH \prod_{i=1}^\ii|u_{i}|^{2n_i}\norm{u}^{2k}\,d\mu_0(u)
\label{eq:strong-cv-expectation}
\end{equation}
for any set of integers $n_i$, with $N:=\sum_i n_i<\ii$ (hence with only finitely many non-zero). Note that the density matrices require to have all the creation operators on the left, but using the canonical commutation relations and the bounds on $\tr[(\cN/T)^{k+\ell}\Gamma_{0,T}]$, we see that the error terms obtained by commuting them are small in the limit $T\to\ii$.

Now we use the factorization properties of quasi-free states and compute, as in the previous section, 
\begin{align*}
\Tr \left[  \prod_{i=1}^\ii \left(\frac{a^\dagger_{i}a_{i}}{T}\right)^{n_i}\Big( \frac{\cN}{T} \Big)^k \Gamma_{0,T} \right]
&= \sum_{\substack{s_j\ge 0\\ \sum s_j =k}}\Tr \left[ \prod_j\left(\frac{a^\dagger_{j}a_j}{T}\right)^{n_j+s_j}\Gamma_{0,T} \right]\\
&= \sum_{\substack{s_j\ge 0\\ \sum s_j =k}}\prod_j\frac{(n_j+s_j)!}{(T(e^{\lambda_j/T}-1))^{n_j+s_j}} .
\end{align*}
Here in the first line we have used
\begin{align} \label{eq:Nk-expansion}
\cN^k =\left(\sum_{i=1}^\infty a_i^\dagger a_i \right)^k = \sum_{\substack{s_j\ge 0\\ \sum s_j =k}} \prod_i(a_i^\dagger a_i)^{s_i}.
\end{align}
Using $T(e^{\lambda_i/T}-1)\to \lambda_i$, we find that
\begin{equation}
\lim_{T\to\ii}\Tr \left[ \prod_{i=1}^\ii \left(\frac{a^\dagger_{i}a_{i}}{T}\right)^{n_i}\Big( \frac{\cN}{T} \Big)^k \Gamma_{0,T} \right]
= \sum_{\substack{s_j\ge 0\\ \sum s_j =k}} \left( \prod_{i=1}^\infty  \frac{(n_i+s_i)!}{\lambda_i^{n_i+s_i}} \right).
\label{eq:cv-nTn-new-Gamma}
\end{equation}
On the other hand, we can use that $\mu_0$ is a Gaussian measure to deduce 
\begin{align} \label{eq:cv-ngamman-new-Gamma}
\int_\gH \prod_{i=1}^\ii|u_{i}|^{2n_i}\norm{u}^{2k}\,d\mu_0(u) &=  \sum_{s_j\ge 0, \sum s_j=k}\int_{\gH} \left( \prod_{i=1}^\infty |u_i|^{2 (n_i+s_i) }\right) \,d\mu_0(u)   \nn\\
&=\sum_{s_j\ge 0, \sum s_j=k} \left( \prod_{i=1}^\infty \frac{(n_i+s_i)!}{\lambda_i^{n_i+s_i}} \right).
\end{align}
Here we have used
$$
|u|^{2k} = \left( \sum_{i=1}^\infty |u_i|^2\right)^k = \sum_{\substack{s_j\ge 0\\ \sum s_j =k}} |u_j|^{2s_j} 
$$
which is analogous to (\ref{eq:Nk-expansion}). This finishes the proof of Lemma~\ref{lem:deFinetti_higher_moments}.\qed

%%%%%%%%%%%%%%%%%%%%%%%%%%%%%%%%%%%%%%%%%%
%%%%%%%%%%%%%%%%%%%%%%%%%%%%%%%%%%%%%%%%%%
% \bibliographystyle{siam}
% \bibliography{biblio_bis}

\end{document}